\documentclass[11pt]{article}
\usepackage{fullpage}
\usepackage[utf8]{inputenc}
\usepackage{amsmath,amsfonts,amsthm,amssymb}
\usepackage{graphicx} 
\usepackage{enumerate}
\usepackage{enumitem}
\usepackage{xspace}
\usepackage[T1]{fontenc}

\usepackage{xcolor}
\usepackage{nameref}
\definecolor{ForestGreen}{rgb}{0.1333,0.5451,0.1333}
\definecolor{DarkRed}{rgb}{0.65,0,0}
\definecolor{Red}{rgb}{1,0,0}
\usepackage[linktocpage=true,
pagebackref=true,colorlinks,
linkcolor=DarkRed,citecolor=ForestGreen,
bookmarks,bookmarksopen,bookmarksnumbered]
{hyperref}
\usepackage{cleveref}

\usepackage{indentfirst}
\usepackage{algpseudocode}
\usepackage{algorithm}  
\usepackage{algorithmicx} 
\usepackage{verbatim}
\usepackage{todonotes}
\usepackage{appendix}
\usepackage{multirow}

\usepackage{thmtools}
\usepackage{thm-restate}

\declaretheorem[numberwithin=section]{theorem}
\declaretheorem[numberlike=theorem]{lemma}

\declaretheorem[numberlike=theorem]{corollary}

\declaretheorem{conjecture} 
\declaretheorem[numberlike=theorem,style=definition]{example} 
\declaretheorem[numberlike=theorem,style=definition]{definition}

\declaretheorem[numberlike=theorem,style=definition]{remark}
\declaretheorem[sibling=conjecture,name=Open Question]{question}

\usepackage{mathrsfs}  

\newcommand{\partition}{\mathbb{P}}
\newcommand{\gridnd}{\mathbb{Z}_{\ge 0}^d}

\Crefname{ALC@unique}{Line}{Lines}

\let\epsilon=\varepsilon
\let\eps=\varepsilon
\usepackage{mdframed}

\usepackage{graphicx} 
\usepackage{array} 
\usepackage{tabularx, booktabs}

\def\vs{{\mathbf{s}}}
\def\vz{{\mathbf{z}}}
\def\vw{{\mathbf{w}}}

\def\radiusbd{{k}}

\newcommand{\unionassign}{\triangleleft}

\definecolor{bleudefrance}{rgb}{0.19, 0.55, 0.91}

\newcommand\fang[1]{\textcolor{green}{[Fang: #1]}}

\title{Packing Compact Subgraphs with Applications to Districting}
\author{Ho-Lin Chen\thanks{
National Taiwan University. Email: holinchen@ntu.edu.tw, math.tmt514@gmail.com, r12922222@ntu.edu.tw. Huang and Chou are currently supported by NSTC grant 114-2222-E-002-004-MY3. Chen is supported by NSTC grant 113-2221-E-002-204-MY3.} 
\and Po-Yu Chou$^\ast$ 
\and Prathamesh Dharangutte\thanks{Department of Computer Science,
Rutgers University. Email: \{prathamesh.d,jg1555\}@rutgers.edu. Dharangutte and Gao would like to
acknowledge NSF support through IIS-2229876, DMS-2220271, CNS2515159, DMS-2311064, and CCF-2118953.} \and 
Jie Gao$^\dag$ \and
Shang-En Huang$^\ast$ \and
Fang-Yi Yu\thanks{Department of Computer Science,
George Mason University, Email: fangyiyu@gmu.edu
}
}
\date{}

\begin{document}

\maketitle

\begin{abstract}
Packing disjoint subgraphs in a given graph is a fundamental problem with many applications. Motivated by political districting, we focus on connected subgraphs that are compact (e.g., having constant radius from a single center vertex) and that satisfy additional composition requirements, such as a minimum population/weight threshold or balanced weight types (e.g., political affiliations). We aim to maximize coverage by disjoint districts that meet these requirements. 

In this work, we present new results that substantially improve the previously known bounds on balanced star districts for planar and minor-free graphs \cite{Dharangutte2025-ry}. In particular, we improve the approximation factor from $O(\log n)$ to $O(1)$ for packing balanced star districts using the exact same algorithm, but with a refined analysis. We also extend the results beyond planar graphs to minor-free graphs and an even broader family of graphs of bounded expansion. Additionally, we obtain an $O(1)$ approximation for packing radius-$k$ districts (with a constant $k$) in planar and apex-minor-free graphs. In order to get a $(1+\eps)$ approximation on the max coverage, we show that this can be achieved if we allow a slight relaxation of the balancedness parameters (by a factor that can be made arbitrarily close to $1$), for bounded radius-$k$ districts on planar and apex-minor-free graphs. 

We show that all of these results can also be obtained if we enforce a minimum weight threshold for each district as the composition requirement, rather than balancedness. We present various results on hardness and hardness of approximation for this variant, by graph and district types. 

\end{abstract}

\newpage
\thispagestyle{empty}
\tableofcontents
\thispagestyle{empty}

\newpage
\setcounter{page}{1}

\section{Introduction}\label{sec:intro}

We consider a general family of districting problems, formulated as packing compact, connected subgraphs, with the objective of maximizing coverage. Given a graph $G=(V, E)$ with vertex set $V$ of $|V|=n$ and edge set $E$.
Each vertex $v$ represents an atomic region with integer \emph{weight} $w$ (e.g., the population counts in the region, or reward).
A district is a \emph{connected} subgraph on a vertex set $S\subseteq V$ in $G$, and two districts must be \emph{vertex-disjoint}. We aim to maximize the total weight 
of the vertices covered by districts. 

Depending on the application scenarios, a district needs to meet additional requirements. One commonly seen requirement is that any two vertices in the same district are not very far, that is, the district is relatively \emph{compact}, or simply not too elongated. This is a desirable property that often appears in discussions of political districting. 
Further, there may be additional requirements on the population composition or budgeting constraints.  We discuss the two dimensions separately.

\paragraph{Compactness} To characterize this requirement, we consider two possible ways of limiting distances between vertices in district $S$: diameter-$k$ or radius-$k$, for a constant $k$.  A district $S$ meets the diameter-$k$ constraint if every two vertices in $S$ are within $k$ hops from each other.  A district $S$ has radius-$k$ if there is a center vertex $v\in S$ and all other vertices of $S$ are within $k$ hops from $v$. 
Further, we can consider the diameter or radius in either the \emph{strong} notion---within the induced subgraph on $S$, or the \emph{weak} notion---in the original graph $G$. Notice that the weak notion is hereditary: for a subgraph $G'$ with weak diameter $k$, any subgraph of $G'$ continues to have weak diameter at most $k$. But this is obviously not always true for the strong diameter or radius. 

\paragraph{Composition} In addition, the districts may also need to satisfy additional properties that involve the composition of the vertices, such as:
\begin{itemize} 
    \item \emph{Balancedness}: 
    in political districting there are two types of weights $w_1(v)$ and $w_2(v)$ for each vertex $v$, with $w(v)=w_1(v)+w_2(v)$. For any vertex subset $S$ and each $i \in \{1, 2\}$ we define $w_i(S)=\sum_{v\in S} w_i(v)$. 
    We call $S$ $c$-balanced if $\min(w_1(S), w_2(S))\geq (w_1(S) + w_2(S))/c$. That is, the two types of weights are comparable in size.

    \item \emph{Weight threshold}: a district is valid if $w(S)\geq B$ for a threshold $B$. This is a natural property that requires all districts to be significant in weight/size. 
Notice that we consider a lower bound type of constraint. The upper bound version is often trivial -- for example, taking all feasible singleton districts is the optimal solution. 

\end{itemize}


\subsection{Motivation}

Packing disjoint subgraphs is a fundamental problem that has many applications. The most notable applications are in political districting, in which connectivity, contiguity and compactness is highly preferred and sometimes enforced by state law~\cite{altman1998traditional}.
Furthermore, the $c$-balanced property guarantees that each political group has a decent amount of representation, which can be viewed as compliance with the Voting Rights Act of 1965~\cite{hebert2010realist}. The weight threshold is another natural property such that each district is significant in size/population. In general, the districting problem can also be applied to the allocation of other types of resources and services~\cite{bucarey2015shape,kalcsics2005towards}. 

Earlier work of~\cite{Dharangutte2025-ry} considered the $c$-balanced districting problem, which is one case of the districting problem defined earlier. 
\cite{Dharangutte2025-ry} provided various hardness results and approximation algorithms for different families of graphs. There are two major directions left open in~\cite{Dharangutte2025-ry}. First, most of the positive results (without additional assumptions on the districts or on very special graphs such as complete graphs or trees) are only for \emph{star} districts, which are (strong) radius-1 subgraphs. It was not clear how the results could be extended to higher, constant values of radius-$k$. Second, and more importantly, one of the most prominent applications in political districting considers planar graphs (e.g., the geographical subdivision into townships and counties). The best approximation factor for packing $c$-balanced star districts in a planar graph in~\cite{Dharangutte2025-ry} is $O(\log n)$ without a matching lower bound. Resolving this issue by either developing a constant-factor approximation or proving a stronger hardness result is a major open question. 


\subsection{Related Work}
Our problem is connected to computational (re)districting for schools and elections, which dates back to the 1960s~\cite{hess1965nonpartisan}. Since then, extensive work (see \cite{becker2020redistricting}) has studied them as an algorithmic problem, typically focusing on specific constraints such as balance or compactness. For balancedness, recent work introduces vote-band metrics~\cite{deford2020computational}, which require a certain fraction of votes to fall within a specified range (e.g., 45-55\%) for competitive elections. Subsequently, \cite{chuang2024drawing,Dharangutte2025-ry} introduce the $c$-balanced district problem. Other research focuses on optimizing balance scores directly~\cite{gillani2023redrawing}. Regarding compactness, one line of work treats it as a transportation cost~\cite{10.1145/3490486.3538271,franklin1973computed,clarke1968operations}, such notions also relate to the fair clustering problem~\cite{Bohm2020-kd, Chierichetti2018-qs, Jia2020-xu,Chhabra2021-nb}. Other research focuses on optimizing compactness scores~\cite{bar2020gerrymandering,jin2017spatial,kim2011optimization,jacobs2018partial} or using Voronoi or power diagrams with some variant of $k$-means~\cite{weaver1963procedure,cohen2017balanced,cohen2018balanced,fryer2011measuring}.
Besides the optimization approach, another popular method uses sampling to generate a distribution over districts and create a collection of district plans for selection~\cite{altman2011bard,cirincione2000assessing,deford2021recombination,najt2019complexity,chen2026balancedspanningtreedistributions}.

Finally, several papers take a fair division approach~\cite{landau2009fair,pegden2017partisan,de2018analysis}. The problem is quite different, however, as fairness in this context is defined concerning parties (types) and the number of seats they would win (i.e., the number of districts where they would have a majority) compared to other districts, rather than the internal composition of individual districts.


Besides its applications in political districting, packing compact subgraphs from a given family, such as paths, cycles, and stars, is a fundamental problem in its own right. Let $\mathcal{H}$ be a family of graphs. A $\mathcal{H}$-packing of a graph $G$ is a set of vertex disjoint subgraphs such that each subgraph is isomorphic to some element of $\mathcal{H}$. These types of problems have been studied in many papers with various hardness and approximation results, 
such as packing stars~\cite{Ning1989-tg,babenko_et_al:LIPIcs.STACS.2011.519,Li2021-hp,Xi2024-oc,Ravelo25}, 
induced stars~\cite{Kelmans1997-fn},
constant sized graphs~\cite{Loebl1993-nb,Caprara2002-iu,letzter2026}, 
cycles and paths~\cite{Hell1988-mj,Huang2021-ro,Ravelo25,BentertFGKLPRSS26},
cliques and bicliques~\cite{Holyer1981-kv,Hell1984-tm,Hell1986-jh}, and
low-diameter spanning trees~\cite{ChuzhoyPT20}. 
Another line of work \cite{abbas1999clustering, dondi2019tractability, DondiMSZ19,DondiL23,MontiS25,zhangpartitioning} explores covering and partitioning graph with $k$-clubs, as a generalization to the clique cover problem. $S \subseteq V$ is a $k$-club in graph $G$ if $S$ induces a subgraph in $G$ with diameter at most $k$ (cliques are $k$-club for $k = 1$). $k$-club with at least $t$ vertices is called a $(t,k)$-club.
The Minimum $k$-club cover problem asks for minimum number of $k$-clubs that cover all vertices in the graph and \cite{DondiMSZ19, zhangpartitioning} explore hardness and approximation algorithms for different values of $k$. The Maximum Disjoint $(t,k)$-club covering problem asks for vertex disjoint $(t,k)$-clubs that cover maximum number of vertices in the graph $G$. \cite{dondi2019tractability, MontiS25} study this problem for $k = 2,3$ on general and bipartite graphs.
Note that none of the above work consider the composition properties of the subgraphs. Thus, the results are only marginally relevant to this paper.  

\subsection{New Results}

In this paper, we resolve both problems left open in~\cite{Dharangutte2025-ry}, significantly expanding the theoretical understanding of this problem. Our results are summarized in \Cref{tab:results}.

\renewcommand{\arraystretch}{1.1}
\begin{table}[t]
\caption{A summary of new approximation results on the threshold and balanced districting problem, $\delta, \epsilon > 0$ are constants. New results are highlighted in bold. }
\label{tab:results}
\centering
\small
\begin{tabularx}{\textwidth}{ >{\centering\arraybackslash}p{3.2cm} | >{\centering\arraybackslash}p{3.8cm} | >{\centering\arraybackslash}X }
\toprule
\multicolumn{1}{c|}{\textbf{Graph Type}} & 
\multicolumn{1}{c|}{\textbf{District Type}} & 
\multicolumn{1}{c}{\textbf{Results}} \\
\toprule
\multirow{3}{*}{Planar} & star (strong radius-$1$) & $O(\log n)$-approx \cite{Dharangutte2025-ry} \\
\cline{2-3}
&  \textbf{strong radius-$k$} & \textbf{$O(1)$-approx (\Cref{thm:k-apex-minor})} \\
\cline{2-3}
& \textbf{weak/strong radius-$k$} & \textbf{$(1+\eps)$-approx, $\delta$-relaxation ~~(\Cref{thm:ptas_relax_aminor_free})} \\
\hline
\multirow{2}{*}{Apex-Minor-Free} & \textbf{strong radius-$k$} & \textbf{$O(1)$-approx (\Cref{thm:k-apex-minor})}\\
\cline{2-3}
& \textbf{weak/strong radius-$k$} & \textbf{$(1+\eps)$-approx, $\delta$-relaxation (\Cref{thm:ptas_relax_aminor_free})} \\
\hline
\multirow{2}{*}{$H$-Minor-Free} & star (strong radius-$1$) & $O(h^2 \log n)$-approx, $h = |H|$ \cite{Dharangutte2025-ry}\\
\cline{2-3}
&  \textbf{star (strong radius-$1$)} & \textbf{$O(1)$-approx (\Cref{thm:star-bounded-expansion})} \\
\hline
Bounded Expansion & \textbf{star (strong radius-$1$)} & \textbf{$O(1)$-approx (\Cref{thm:star-bounded-expansion})} \\
\bottomrule
\end{tabularx}
\end{table}

\renewcommand{\arraystretch}{1}

\paragraph{$O(1)$-approximation} Compared to the $O(\log n)$-approximation of packing star districts on a planar graph in~\cite{Dharangutte2025-ry}, we improved the results in two directions.
First, we show an algorithm (\Cref{thm:star-bounded-expansion}) with a \emph{constant} approximation factor for packing $c$-balanced \emph{star} districts in graphs of \emph{bounded expansion}, which is a concept of graph family that are `everywhere sparse' introduced by Ne\v{s}et\v{r}il and Ossona de Mendez~\cite{Nesetril2008-lt, Nesetril2008-vp, Nesetril2008-af}. Graphs with bounded expansion include minor-free graphs (and planar graphs) and graphs of bounded degree, as well as graphs that appear in complex network models such as the configuration model and the Chung–Lu model~\cite{DEMAINE2019199}.
This improves the approximation factor from $O(\log n)$ to $O(1)$ for a much bigger family of graphs. 
Second, we show a \emph{constant} approximation for packing $c$-balanced or threshold districts of \emph{bounded strong radius} on \emph{planar graphs and apex-minor-free} graphs\footnote{An apex graph is a graph that can be made planar by the removal of a single vertex and apex-minor-free graphs are a minor-closed family of graphs that has an apex graph as one of its forbidden minors.} (\Cref{thm:k-apex-minor}). This advances the previous results by improving the approximation factor from $O(\log n)$ to $O(1)$ and supporting packing of districts of bounded strong radius. Both results also apply to the threshold districts (instead of $c$-balanced ones for the composition requirement).

\paragraph{$(1+\eps)$-approximation} Next, we wonder what can be said if we want a $(1+\eps)$ approximation for the district packing problem. 
We show that PTAS can be obtained for packing \emph{bounded weak radius} districts on \emph{planar and apex-minor-free graphs}, if we allow a small relaxation of $c$-balanced property to $c(1+\delta)$-balanced, for any $\delta>0$ (\Cref{thm:ptas_relax_aminor_free}). The same holds for the weight-threshold version if we relax the threshold $B$ to $B(1-\delta)$ (\Cref{thm:ptas_relax_aminor_free}).

\paragraph{Hardness results} Finally, we complement the positive results with hardness results in \Cref{sec:hardness}.
For the majority of the positive results, we consider compactness as defined by a bounded radius $k$. For maximum packing of subgraphs of bounded \emph{diameter} (both weak and strong) on general graphs, we show that even finding \emph{one} district with weight threshold is NP-hard (\Cref{thm:hardness-diam-single}) by reducing it from the max clique problem.
Additionally, the separation oracle for the LP formulation relies on finding a violating district, which itself is a feasible district; implying that implementing the separation oracle for districts of bounded diameter in general graphs is NP-hard. 
For maximum packing of districts with a weight threshold, we show that it is NP-hard by reduction from the maximum independent set problem (\Cref{thm:hardness-max-packing}). Specifically, we show it is NP-hard to approximate to a factor better than $n^{1/2-\delta}$ for $\delta > 0$ on general graphs, and the maximum packing problem remains NP-hard even on planar graphs. This hardness holds even when we require the subgraphs to have bounded radius, with $k=1$ (stars). The problem is trivial for a complete graph but remains NP-hard on a tree by a reduction from the Knapsack problem.

Last, the $c$-balanced property considers two types of weights; if three or more weights are considered, the problem of packing compact graphs is hard to approximate within any factor (see \Cref{thm:three-colors}).

\medskip

We will now formally define the problems and then provide an overview of the challenges and new techniques to obtain the results.

\section{Preliminary}\label{sec:framework}


\subsection{Problem Statement}\label{sec:model}


Consider a connected undirected graph  $G=(V,E)$ with $n=|V|$ vertices and $m=|E|$ edges. We define a general district packing problem. The input graph is associated with an \emph{objective weight} function $w: V \rightarrow \mathbb{Z}_{\geq 0}$ and potentially a list of \emph{feature weight} functions $w_1, \dots, w_d: V \rightarrow \mathbb{Z}_{\geq 0}$. These feature weights can be independent of each other and independent of objective weights.
A \textit{district} is a subset of vertices $T \subseteq V$ such that the induced subgraph $G[T]$ is connected and satisfies specific validity constraints: \emph{compactness} (constraints on topology) and \emph{composition} (constraints on feature weights). A \emph{(partial) districting} solution is a collection of vertex-disjoint districts $\mathcal{T} = \{T_1, \dots\}$. We remark that finding a complete partition of the graph into vertex-disjoint districts (i.e., $\bigcup_{T \in \mathcal{T}} T=V$) may not be possible; rather, we seek a packing of disjoint districts with a maximum total objective weight.


For the compactness constraints, besides connectivity, we consider radius or diameter bounds.  Specifically, the distance $d_G(u,v)$ between vertices $u$ and $v$ in $G$ is the number of edges on the shortest path between $u$ and $v$. 
 A district $T$ has bounded \emph{strong radius} $\radiusbd$ if there exists a vertex $\gamma \in T$ such that for all other vertices $u \in T$ , $d_{G[T]}(u,\gamma) \leq \radiusbd$. Similarly, $T$ has bounded \emph{weak radius} $k$ if there exists a vertex $\gamma \in V$ such that for all other vertices $u\in T$, $d_G(u,\gamma) \leq k$. A districting $\mathcal{T}$ satisfies strong (or weak) radius bound $\radiusbd$ if every $T\in \mathcal{T}$ has bounded strong (or weak) radius $\radiusbd$.
A district $T$ has bounded \emph{strong diameter} if all pairs of vertices $u, v$ in $T$ have distance in the induced subgraph at most $k$, i.e., $d_{G[T]}(u,v) \leq k$. It has bounded \emph{weak diameter} if for every pair of vertices $u, v \in T$, their distance in the original graph $G$ is at most $k$, i.e., $d_G(u,v) \le k$.

Composition constraints impose limits on the feature weights of each district. First, we can extend our weight to a single district $T$ and a districting solution $\mathcal{T}$.  For any weight function (representing either the objective $w$ or a feature $w_i$), the district weight is $w(T) = \sum_{v \in T} w(v)$, and the total weight of the districting solution $\mathcal{T}$ is $w(\mathcal{T}) = \sum_{T \in \mathcal{T}} w(T)$.
We consider two types of composition requirements:


\paragraph{Balanced Districting.} 
Given feature weights $w_1, w_2$ and a balancedness parameter $c \geq 2$, a district $T$ is \emph{$c$-balanced}~\cite{Dharangutte2025-ry} if:
\begin{equation}\label{eq:def_balanced}
    \min\{w_1(T), w_2(T)\} \geq \frac{1}{c}\big(w_1(T) + w_2(T)\big).
\end{equation}
A districting solution $\mathcal{T}$ is $c$-balanced if every district $T \in \mathcal{T}$ satisfies Eq.~\eqref{eq:def_balanced}. 
For bicriteria approximation, we may also refer to a \emph{$\delta$-relaxed $c$-balanced} condition (where $0 < \delta < 1$):
\begin{equation}\label{eq:def_relax_balanced}
    \min\{w_1(T), w_2(T)\} \geq \frac{1}{c(1+\delta)}\big(w_1(T) + w_2(T)\big).
\end{equation}
\paragraph{Threshold Districting.} 
Given a feature weight $w_1$ and a threshold parameter $B \geq 0$, a district $T$ is a \emph{$B$-threshold} district if:
\begin{equation}\label{eq:def:threshold}
    w_1(T) \geq B.
\end{equation}
A districting solution $\mathcal{T}$ is $B$-threshold if every district $T \in \mathcal{T}$ satisfies Eq.~\eqref{eq:def:threshold}.

Similarly, a \emph{$\delta$-relaxed $B$-threshold} district is defined by the condition:
\begin{equation}\label{eq:def_relax_threshold}
    w_1(T) \geq B(1-\delta).
\end{equation}
We will use $S,T$ or $U$ to denote districts and $\mathcal{S}$ for the set of valid districts (e.g., the set of $c$-balanced districts with strong radius bound $\radiusbd$).

\paragraph{Approximate and Bicriteria Solutions.}
Given a graph $G$, weights $w, \{w_i\}$, and constraint parameters ($\radiusbd$ and $c$ or $B$), our goal is to find a districting solution $\mathcal{T}$ that maximizes the total objective weight $w(\mathcal{T})$ and every district $T \in \mathcal{T}$ satisfies the (compactness and composition) constraints.

We say a districting solution $\mathcal{T}$ is an \emph{$\alpha$-approximation} if every district in $\mathcal{T}$ satisfies the exact constraints and
$$ \alpha \cdot w(\mathcal{T}) \geq w(\mathcal{T}_{\text{OPT}}), $$
where $\mathcal{T}_{\text{OPT}}$ is an optimal solution satisfying the exact constraints.
We also consider \emph{bicriteria $(\alpha, \delta)$-approximations}. In this setting, the solution $\mathcal{T}$ satisfies the weight approximation guarantee ($\alpha \cdot w(\mathcal{T}) \geq w(\mathcal{T}_{\text{OPT}})$) but is permitted to satisfy the $\delta$-relaxed composition constraints.

For instance, in the $c$-balanced districting problem with strong radius bound $\radiusbd$:
\begin{itemize}
    \item An {$\alpha$-approximation} requires that each district in $\mathcal{T}$ is strictly $c$-balanced and has strong radius at most $\radiusbd$.
    \item A {bicriteria $(\alpha, \delta)$-approximation} allows the districts in $\mathcal{T}$ to be $\delta$-relaxed $c$-balanced (i.e., balanced with parameter $c(1+\delta)$), while still comparing the solution against the strictly $c$-balanced optimum $\mathcal{T}_{\text{OPT}}$.
\end{itemize}

\paragraph*{Graph Types.} A graph $G=(V, E)$ is \emph{planar} if there exists an embedding of all vertices to the Euclidean plane such that all edges can be drawn without intersections other than the endpoints.
A \emph{face} of a planar embedding is a connected region separated by the embedded edges. $G$ is said to be \emph{outerplanar} if there exists an embedding of $G$ such that there is a face containing all vertices. Often, this face is assumed to be the outer face. A graph $H$ is said to be a \emph{minor} of $G$ if $H$ is isomorphic to the graph obtained by a sequence of vertex deletions, edge contractions, and edge deletions from $G$. We say that $G$ is \emph{$H$-minor-free} if $G$ does not have $H$ as its minor. A graph is an apex graph if there is a vertex whose removal results in a planar graph. A graph $G$ is apex-minor free if it does not have some apex graph $H$ as a minor. See~\cite{Eppstein2000-ry,Demaine04equivalence} for more discussions on apex-minor-free graphs. We now formally define bounded expansion graphs.

\begin{definition}
\label{def:bounded_exp_graph}
Let $G=(V, E)$ be an undirected (multi-)graph and let $f:\mathbb{Z}_{\ge 0}\to\mathbb{Z}_{\ge 0}$ be a non-decreasing function. We say that $G$ belongs to the class of \emph{$f$-bounded expansion graphs} if the following holds.
Let $(V_1, V_2, \ldots, V_{n'})$ be a partition of $V$, where each part $V_i$ induces a connected graph.
Let $k$ be the maximum diameter of the induced subgraphs $G[V_i]$ for all $i$. Consider the contracted graph $H$ from $G$ where each part $V_i$ is contracted to a single vertex. Then, for any subgraph of $H'\subseteq H$, the density of $H'$ satisfies $\frac{2|E(H')|}{|V(H')|} \le f(k)$.
\end{definition}
At a high-level, a bounded expansion graph $G$ keeps a bounded density (and thus bounded degeneracy) if one contracts bounded diameter subgraphs all at once in $G$.

\subsection{Algorithmic Framework}

We first examine the districting problem from a bird's-eye view and explain how the technical tools introduced in this work can be applied more generally.

\paragraph{Linear Program Formulation}
Recall that a candidate district is represented by a set of vertices $S\subseteq V$ as defined in~\cref{sec:model}. We use $\mathcal{S}$ as the set of all candidate districts that meet the requirements. This set $\mathcal{S}$ can be of potentially exponential size. 
For each candidate district $S\in \mathcal{S}$, we define a variable $x_S\in \{0, 1\}$ indicating whether or not this district is chosen. Due to the disjointness of districts, among all the districts that cover the same vertex $v$, only one of them can take a value of $1$. Denote by $w(S)$ the weight of $S$. The integer program is defined as 
\begin{equation}\label{eq:lp-main}
\begin{aligned}
\text{maximize}\ &\sum_{S} w(S) x_S \\
\text{subject to}\ &\forall v\in V, \sum_{S\ni v} x_S \le 1\\
\ &\forall S, x_S\in \{0, 1\}
\end{aligned}
\end{equation}

To solve an integer solution, we use two steps. First, solve for an (approximate) solution to the LP problem. Then round the fractional solution to an integer solution. We discuss the two problems separately. 
To relax this integer program to a fractional solution, we allow $x_S$ to take real values in $[0, 1]$ and define $x'_S=x_S\cdot w(S)$, $x'_S \in [0, w(S)]$. Rewriting the LP using $x'_S$, the primal and dual LP problems are as follows.

\begin{mdframed}
\textbf{Primal:} Maximizing $\sum_S x'_S$, subject to $\sum_{S: v\in S} \frac{x'_S}{w(S)} \le 1$ for all $v\in V$; $x'_S\ge 0$ for all $S$. \\
\textbf{Dual:} Minimizing $\sum_v y'_v$, subject to $\sum_{v\in S} \frac{y'_S}{w(S)} \ge 1$ for all $S$; $y'_v\ge 0$ for all $v$.
\end{mdframed}

Solving the LP problem is straightforward if we are given a set of polynomially many valid districts. One scenario in which this could happen is when a districting already exists on $G$ and, for a new districting solution, only minor changes to the previous solution are acceptable.  
In general, the number of potential districts could be exponentially large, even if we consider only stars (e.g., under the compactness restriction). 
We cannot explicitly write the variables and constraints in the LP and must rely on a separation oracle. For the dual LP problem, a \textit{strongly violating dual constraint} is a district $S$ such that $S$ is valid and for the dual variables values of vertices in $S$ and $0 < \varepsilon < 1$, $\sum_{v \in S} y'_v < (1-\varepsilon) w(S)$. We formulate the task of implementing a separation oracle as follows:

\begin{mdframed}
\textbf{Input:}
Graph $G$, dual variables $\{y'_v\}$, weight function $w: V \rightarrow Z_{\geq 0}$, implicit set $\mathcal{S}$ of all valid districts.
\\
\textbf{Goal:} Either returns a violating district $S \in \mathcal{S}$ such that its weight $w(S)$ is at least half to the maximum weight among all violating districts, and 
\begin{equation}\label{eq:violating-district}
    \sum_{v\in S} y'_v < (1-\epsilon/2) w(S); 
\end{equation}
    or reports that \emph{all} districts $S\in \mathcal{S}$ satisfy $\sum_{v \in S} y'_v \ge (1-\epsilon) w(S)$.
\end{mdframed}

\paragraph{Randomized Rounding}

Once we have an (approximate) fractional solution to the LP, we use the same randomized rounding as in~\cite{Dharangutte2025-ry} to find an integer solution $I$: sort all districts with non-zero $x$-values in decreasing order of their weights; in this order a district $S$ is selected with probability proportional to $x_S$, if it does not overlap with any districts already selected in the integer solution. This rounding procedure introduces another approximation factor of $O(\tau)$, where $\tau$ is called the \emph{correlation ratio} and bounds the correlation term $\sum_{A,B \in \mathcal{S}, A \cap B\neq \emptyset} x_A x_B\leq \tau \sum_{S\in \mathcal{S}}x_S$. 

We state this lemma from \cite{Dharangutte2025-ry} for completeness. Suppose $\mathcal{S}_{\geq \delta}$ is the set of districts from the fractional LP solution that have $x$-value at least $\delta$ and $w(I)$ is the total weight of the integer solution $I$ after rounding. The expectation is on the random coin flips in the rounding process. 
\begin{lemma}[Lemma B.5 in~\cite{Dharangutte2025-ry}]\label{lem:rounding}
    Suppose that there exists a $\tau\in \mathbb{R}_{>0}$ such that for all $\delta>0$, $$ \sum_{A,B \in \mathcal{S}_{\geq \delta}, A \cap B\neq \emptyset} x_A x_B\leq (\tau/2) \cdot \sum_{S\in \mathcal{S}_{\geq \delta}}x_S,$$
    then $E[w(I)]\geq \frac{1}{2\tau} \sum_S w(S)\cdot x_S$.
\end{lemma}


\section{Technical Overview}\label{sec:technical-details}


This section summarizes the main algorithmic ideas and the core technical challenges underlying Sections~\ref {sec:rounding} and~\ref{sec:minorfree}.
The new progress requires addressing two major challenges. First, we need to identify the candidate districts that meet the specified requirements (connectivity, compactness, balancedness, or weight threshold). Next, we choose a disjoint subset of the candidate districts with maximum total weight. The first issue is related to \emph{feasibility}, and the second issue regards \emph{disjointness} or \emph{packing}. 

At a high level, our positive results rely on two complementary approaches. Section~\ref{sec:rounding} develops a refined analysis of a linear programming (LP) relaxation combined with randomized rounding, yielding constant-factor approximations on minor-free and bounded-expansion graph classes. Section~\ref{sec:minorfree} strengthens the result for apex-minor-free graphs by combining Baker’s layering technique with a new dynamic programming (DP) framework on bounded-treewidth graphs, leading to a PTAS under relaxed parameters.

\subsection{Refined Rounding Analysis (Section~\ref{sec:rounding})}

Recall that $x_S$ is the primal variable in the LP formulation for a district $S$, $x_S\in [0,1]$. We consider districts of strong radius $k$. The LP framework allows addressing feasibility and packing in an orthogonal manner. The separation oracle helps to find a feasible district that substantially improves the LP solution, and the randomized rounding manages packing.

The performance of the randomized rounding procedure is governed by a second-order correlation term
$
\sum_{A \cap B \neq \emptyset} x_A x_B,
$
which measures the extent to which overlapping districts interfere during rounding. 
Given a graph $G$ and  $\mathcal{S}$ (a collection of connected subgraphs of $G$), we say that $G, \mathcal{S}$ has \emph{correlation ratio} $\tau$ if  $\sum_{S: v\in S} x_S\le 1$, and $\sum_{A \cap B \neq \emptyset} x_A x_B
\le
\tau \sum_{S} x_S,$ for all  $v$ and $\{x_S\}$. 
\Cref{lem:rounding} gives an $O(\tau)$ approximation guarantee if we can establish any bound on $\tau$. In particular, in \Cref{thm:main} we show that if $G$ is $H$-minor-free for some $H$ of size $h$, then $\tau$ can be bounded by a function $q(h, k)$ that depend solely on $h$ and $k$.
Establishing this bound yields a constant-factor approximation.

The main technical difficulty is that intersecting districts may overlap in highly structured ways. Even when each district has a strong radius $k$, a fixed district $A$ could, a priori, intersect with many other districts $B$.
A na\"{i}ve charging argument could therefore lead to an $O(n)$ blow-up.\footnote{For example, consider $A$ being a star of $\Omega(n)$ vertices, with each leaf intersecting distinct districts. In this case, the terms $x_A$ involved in the correlation sum are at least $\Omega(n)x_A$. This shows some districts will be charged for $\Omega(n)$ times in a na\"{i}ve charging argument.}

To overcome this, we exploit the structural sparsity of minor-free graphs. Specifically, an $H$-minor-free graph has bounded \emph{degeneracy}---there is an orientation of the edges such that the out-degree of any vertex is bounded by $O(h\sqrt{\log h})$~\cite{Thomason01,AlonKS23}. Using this property, we design a probabilistic charging scheme based on a sequence of random edge orientations and contractions. For each intersecting pair $(A,B)$, we identify a carefully structured ``intersecting pivot'' along a canonical path between their centers. We then show that with probability $\Omega_{h,k}(1)$, the pair can be charged to one endpoint district in a way that ensures:

\begin{enumerate}
    \item Each intersecting pair is charged with constant probability.
    \item Each district receives only $O_{h,k}(1)$ total charge.
\end{enumerate}

This argument combines bounded degeneracy, controlled path contractions, and a forward-trajectory encoding that limits the number of distinct pivots that can be associated with a fixed district. The result is a constant bound on the correlation term, completing the rounding analysis. 

The analysis here is independent of how the LP fractional solution is computed and thus can be combined with any algorithm or implementation that produces an (approximate) fractional solution. 

A similar argument extends to bounded-expansion graphs (\Cref{subsec:expansion}), since such graphs have constant `density' everywhere and also enjoy bounded degeneracy. With a slightly modified contraction, the correlation term is bounded by a factor that depends only on the density parameter.  
We remark that the randomized rounding algorithm we use is fairly standard~\cite{10.1145/73393.73394,ChanH12}, but the analysis of its performance is new and may find further applications.

The new analysis of randomized rounding, combined with the separation oracle for star districts in~\cite{Dharangutte2025-ry}, immediately yields an $O(1)$-approximation algorithm for packing $c$-balanced or $B$-threshold star districts in planar, minor-free, and bounded-expansion graphs.

In order to handle districts of a general radius $k$ with $k>1$, we also need to extend the separation oracle for star districts in~\cite{Dharangutte2025-ry} to accommodate a general radius $k$ (\Cref{sec:separation}). To do that, we formulate a vector-valued subset sum problem on a connected compact subgraph. Similar to the standard subset sum, we trim the solution space by keeping a succinct representation of the subsets that approximate every possible solution with a multiplicative factor and maximize a linear score. This new framework can be applied to both $c$-balanced districts and threshold districts, as well as to a generic connected knapsack problem (see~\cite{Dey2024-jg}). 
For the connected subgraph sum problem, we present an FPTAS on graphs of \emph{bounded treewidth}. On a planar graph or an apex-minor-free graph, the $k$-hop neighborhood of any center vertex has bounded treewidth~\cite{Eppstein2000-ry,Bodlaender1998-ma}. Thus, we can enumerate every possible center vertex and apply the separation oracle in each bounded radius-$k$ neighborhood.
This provides an efficient separation oracle for planar and apex minor-free graphs.


\subsection{PTAS for Apex Minor-Free Graphs (Section~\ref{sec:minorfree})}

Section~\ref{sec:minorfree} improves the approximation guarantee for apex-minor-free graphs using a different paradigm with two elements: the Baker's layering approach~\cite{Baker1994-ih} for planar and apex minor-free graphs~\cite{Demaine2005-sj} and a PTAS for \emph{subgraph packing} for graphs of bounded treewidth. 

We explain the main idea for planar graphs. The Baker's approach considers a BFS from an arbitrary root and removes the vertices on the layers $=i \mod t$, for some $0\leq i\leq t-1$. With a fixed $i$, this will partition the graph into connected components, each with at most $t$ consecutive layers of the BFS. Thus, each piece is $t$-outerplanar and thus has treewidth at most $3t-1$~\cite{Bodlaender1998-ma}. 

Suppose, for now, we have a magic algorithm $\mathcal{A}$ that solves compact subgraph packing in a graph of bounded treewidth. We apply algorithm $\mathcal{A}$ for each piece, such that all obtained districts from the connected components combined form a feasible packing solution for the original graph $G$. If we try this for all possible values of $i$ between $0$ and $t-1$, and take the maximum solution, this is a $(1+\eps)$ approximation to the optimal with $t=O(\radiusbd/\eps)$. The reason is that in an optimal solution, only the subgraphs that are `influenced' by the removed layers are destroyed. These vertices are within $O(\radiusbd)$ hops from the removed layers since each compact subgraph has radius $\radiusbd$. 
By the pigeonhole principle, at least one of the $t$ candidate solutions keeps all but $O(\radiusbd/t)=O(\eps)$ fraction of the optimal weight. 

The central technical task is therefore to find the algorithm $\mathcal{A}$ that solves the problem of packing compact districts on bounded-treewidth graphs. Unlike the separation oracle problem (which asks for a \emph{single} feasible district), packing requires simultaneously selecting \emph{multiple} disjoint connected subgraphs subject to either balancedness or threshold constraints. This introduces a significantly more complex state space.

We develop a dynamic programming scheme over a tree decomposition that encodes partial district configurations within each bag. The DP must simultaneously track the connectivity information of partially formed districts and the cumulated weight (for threshold constraints) or the weight balance (for $c$-balanced districts). Our DP mechanism has to relax the parameter $c$ slightly (for the $c$-balanced districts) or the weight threshold $B$ (for the threshold districts) by a factor. 

\section{$O(1)$-approximation on Planar Graphs and Beyond
}\label{sec:rounding}

\subsection{Constant Ratio for Minor-Free Graphs}

In this section, we briefly describe the linear programming approach for the district packing problem, together with the randomized rounding framework based on the support variables of an optimal linear program solution.
The approximation guarantee of this framework relies on bounding certain second-order terms involving the support variables.
In particular, throughout this section we aim to prove the following key technical \Cref{thm:main} for minor-free graphs and \Cref{thm:bdd-expansion} for bounded expansion graphs.

\begin{restatable}{lemma}{correlationterm}
\label{thm:main}
For all $k\ge 1$ and $h\ge 1$, there exists $q(h, k) = h^{O(k^2)}$ such that the following holds.
Let $G$ be an $H$-minor-free graph with $|H|\le h$, $\mathcal{S}$ be the collection of connected subgraphs with strong radius $k$, and $x: V\to \mathbb{R}_{\ge0}$ is a nonnegative weight function satisfying $\sum_{S\in \mathcal{S}: S\ni v} x_S\le 1$ for all $v$ for all $v$.  Then,
\[
\sum_{A,B \in \mathcal{S} \text{ and } A \cap B \neq \emptyset} x_A x_B
\;\le\;
q(h,k) \cdot \sum_{S \in \mathcal{S}} x_S .
\]  
\end{restatable}


Combining the above lemma with the linear programming formulation and the randomized rounding framework yields the following results.

\begin{theorem}\label{thm:k-apex-minor}
There exists a polynomial-time algorithm that computes an $O(1)$-approximate solution for packing $c$-balanced or threshold districts with \emph{bounded strong radius} on \emph{planar graphs} and \emph{apex-minor-free graphs}.
\end{theorem}

\begin{theorem}\label{thm:star-bounded-expansion}
There exists a polynomial-time algorithm that computes an $O(1)$-approximate solution for packing $c$-balanced or threshold \emph{star} districts on \emph{$H$-minor-free graphs} and \emph{bounded-expansion graphs}.
\end{theorem}

We prove \Cref{thm:k-apex-minor} and \Cref{thm:star-bounded-expansion} at the end of this section.
The high level approach of proving \Cref{thm:main} is to partition the pairs of intersecting districts into different cases, and give the upper bound to each of the case. To start with the simplest case, we first introduce the following notions of \emph{center vertices} and \emph{intersecting pivots}.

\paragraph{Center Vertices and Intersecting Pivots.}
For each district $S \in \mathcal{S}$, we designate a \emph{center} vertex $c_S \in S$ to the district, where every vertex in $S$ is within distance $k$ from $c_S$.
If there are multiple choices, $c_S$ can be chosen arbirarily. Note that since $S$ has a strong radius at most $k$, there is at least one center vertex in $S$.

For each pair of districts $(A, B)$ that intersect, we further define the following terms.

\begin{definition}
Let $A, B\in \mathcal{S}$ be two districts with $A\cap B\neq \emptyset$. Let $P_{A, B}$ be an arbitrarily chosen (undirected) path that satisfies the following conditions: 
\begin{enumerate}[itemsep=0pt]
    \item $P_{A, B}$ is a path connecting $c_A$ and $c_B$. If $c_A=c_B$ then $P_{A, B}$ is a single vertex.
    \item There exists an \emph{intersecting pivot} $i_{A, B}\in P_{A, B} \cap A\cap B$.
    \item The length of the subpath between $c_A$ and $i_{A, B}$, and the length of the subpath between $c_B$ and $i_{A, B}$ are both at most $k$. 
    \item All vertices on the path between $c_A$ and $i_{A, B}$ are in $A$, and all vertices on the path between $i_{A, B}$ and $c_B$ are in $B$.
\end{enumerate}
We require $P_{A, B}=P_{B, A}$ and $i_{A, B}=i_{B, A}$.
\end{definition}

Now, with the definition of the intersecting pivot, we are able to partition the set of intersecting district pairs $(A, B)$ into two types --- an ``easy'' case where $i_{A, B} \in \{c_A, c_B\}$, or otherwise a ``non-trivial'' case.


For the easy case, without loss of generality, we may assume  $i_{A, B}=c_A$. In this case, we bound the total correlation using the primal constraints:
\[\sum_{A\cap B\neq \emptyset\text{ and } i_{A, B}=c_A} x_A x_B \leq \sum_{A\in \mathcal{S}} x_A \sum_{B \ni c_A} x_B  \leq \sum_{S \in \mathcal{S}} x_S\]


Therefore, it remains to bound the correlation sum for all intersecting districts $(A, B)$, where $c_A, c_B$ and $i_{A, B}$ are all distinct vertices.

\subsection{A Probabilistic Method for Challenging Case}
The analysis of the easy case does not simply generalize to the non-trivial case where $i_{A, B}, c_A$ and $c_B$ are all distinct. One potentially convincing reason is that if we fix the district $A$, all the intersecting pivots $i_{A, B}$ might be distinct, which implies that the sum over all intersecting districts $\sum_{B: A\cap B\neq \emptyset} x_Ax_B$ could be as large as $|A|\cdot x_A = \Omega(n) \cdot x_A$ in the worst case.

Fortunately, one can utilize the bounded degeneracy property of $H$-minor-free graphs for a fixed minor $H$ of size $h$.
A \emph{$d$-degenerate graph} is an undirected graph in which every subgraph has at least one vertex of degree at most $d$. The \emph{degeneracy} of a graph is the smallest value $d$ for which it is $d$-degenerate.
The following fact bounds the degeneracy of $H$-minor-free graphs.

\begin{lemma}[\cite{Thomason01, AlonKS23}]
Let $G=(V, E)$ be an $H$-minor free graph for a fixed graph $H$ of $h$ vertices, then the degeneracy of $G$ is $d=O(h \sqrt{\log h})$. 
\end{lemma}

With this property, we are able to design a randomized charging scheme, such that for each pair of intersecting districts $(A, B)$, with probability $\Omega_{h, k}(1)$, the cost $x_Ax_B$ will be charged to either $x_A$ or $x_B$ and simultaneously each $x_A$ or $x_B$ will be charged by at most $O_{h, k}(1)$ times.  
This completes the proof by the standard probabilistic method.

We summarize the above discussion into the following lemma, which directly leads to \Cref{thm:main}.

\begin{lemma}\label{lemma:randomized-charging-scheme}
Let $G=(V, E)$ be an $H$-minor-free graph (where $H$ has $h$ vertices) and $\mathcal{S}$ be a set of subgraphs with strong radius $k$.
Then, there exists a randomized algorithm $\mathcal{R}$ that takes $(G, \mathcal{S}, \{x_S\})$ as input, and output a list $\mathcal{O}$ of ordered tuples $(A, B, j) \in \mathcal{S}\times \mathcal{S}\times V$, such that the following holds.
\begin{enumerate}[itemsep=0pt]
    \item[(1)] $j\in B$.
    \item[(2)] For each district $A\in \mathcal{S}$, the set $J_A:= \{ j \ |\  (A, B, j) \in \mathcal{O}\}$ contains at most $d^{k+1} {\binom{2k}{k}}$ distinct vertices.
    \item[(3)] For any intersecting pair $(A, B)$, with probability at least $1/(d+1)^{k^2}$, either $(A, B, j)\in \mathcal{O}$ or $(B, A, j')\in\mathcal{O}$, for some $j\in B$ or some $j'\in A$.
\end{enumerate}
\end{lemma}
Before proving \Cref{lemma:randomized-charging-scheme}, we first show that \Cref{lemma:randomized-charging-scheme} implies the correctness of \Cref{thm:main}, the main technical lemma of this section.

\begin{proof}[Proof of \Cref{thm:main}]
By item (3) in \Cref{lemma:randomized-charging-scheme}, we obtain the expected bound:
\[
\operatorname*{\mathbf{E}}_{\mathcal{R}}\left[\sum_{(A, B, j)\in \mathcal{O}} x_Ax_B\right] \ge \frac{1}{(d+1)^{k^2}} \sum_{A\cap B\neq \emptyset} x_Ax_B .
\]
On the other hand, by items (1) and (2) we have \[
\sum_{(A, B, j)\in \mathcal{O}} x_Ax_B = 
\sum_{A\in\mathcal{S}} x_A\sum_{j \in J_A} \left(\sum_{B:\ (A, B, j)\in\mathcal{O}} x_B \right) \le 
d^{k+1} {\textstyle \binom{2k}{k}} \sum_{S\in \mathcal{S}} x_S.
\]
By the probabilistic method argument, there must exist a set of intersecting pairs $\mathcal{O}'$ that satisfies both inequalities,
which completes the proof of \Cref{thm:main} with $q(h, k) = (d+1)^{k^2}d^{k+1}\binom{2k}{k} = h^{O(k^2)} $.
\end{proof}

The rest of the section devotes to proving  \Cref{lemma:randomized-charging-scheme}. We first describe the randomized algorithm $\mathcal{R}$ below. 

\paragraph{The Algorithm $\mathcal{R}$.} 
The algorithm consists of two parts: in the first part, the algorithm randomly produces a sequence of minors of $G$: $G_0=G, G_1, \ldots, G_k$ via a sequence of random contractions, where $k$ is the radius/diameter constraint.
In the second part, the algorithm considers each intersecting district pair $(A, B)$, obtains some vertex $j\in V$, and decides if the tuple $(A, B, j)$ will be added to the output list and counted toward the correlation sum.
Note that due to asymmetry, $(A, B)$ and $(B, A)$ are considered separately.

\paragraph{Part 1:}
The algorithm initializes a copy $G_0\gets G$ of the input graph.
Run the following procedure in $k$ rounds.
In each round $r=1, 2, \ldots, k$, the algorithm computes an arbitrary orientation $\vec{G}_{r-1}$ from $G_{r-1}$, where each vertex has out-degree at most $d$.
The algorithm then assigns a randomly chosen neighbor (or not choosing at all) for each vertex, denoted as the variable $\textit{pick}_r[v] \in V\cup \{\bot\}$. Specifically, for each vertex $v\in V$, each out-going neighbor $v'\in out(v)$ will be assigned to $\textit{pick}_r[v]$ with probability $\frac{1}{d+1}$, and for the remaining probability $1-|out(v)|/(d+1)$, the algorithm chooses nothing (sets $\textit{pick}_r[v] \leftarrow \bot$) for $v$.

Finally, at the end of each round, the algorithm forms the (undirected) contracted graph $G_r$ by contracting $v$ toward $\textit{pick}_r[v]$ all at once. The contraction naturally induces a partition of $V(G_{r-1})$, where each part is contracted to a single vertex and there is at most one vertex $v$ in the part with $\textit{pick}_r[v]=\bot$. If such vertex $v$ exists, we name the contracted vertex in $G_r$ to be $v$. Otherwise, we do not need this vertex anymore so they can be completely removed from $G_r$. We emphasize that this naming convention is important to our analysis.

\paragraph{Part 2:} 
The algorithm computes an orientation $\vec{G}_k$ from the latest contracted graph $G_k$.
For each pair of intersecting district $(A, B)$ with $i_{A, B}$ different from $c_A$ and $c_B$,
the algorithm checks if the following event $E_{A, B}$ holds:
\begin{enumerate}[itemsep=0pt]
    \item The vertices $c_A, c_B$ are never contracted. That is, $\textit{pick}_r[c_A]=\textit{pick}_r[c_B]=\bot$ for all $r$.
    \item $\textit{pick}_r[i_{A, B}]=\bot$ for all $r=1, 2, \ldots, k-1$. Note that we do not restrict $\textit{pick}_k[i_{A, B}]$.
    \item For every other vertex $v$ on $P_{A, B}$ except $c_A, c_B, i_{A, B}$, whenever $\textit{pick}_r[v]\neq \bot$, we must have $\textit{pick}_r[v]\in P_{A, B}$. In other words, if a vertex on $P_{A, B}$ participates in a contraction, the vertex contracts ``along the path'' and shortens the path.
    \item There are three possible values for $\textit{pick}_k[i_{A, B}]$ (contract to the left, contract to the right, or stay.)
    Define $j'$ and $j$ to be the two vertices on $P_{A, B}$ closest to $i_{A, B}$ that are not contracted in the end, such that $j'\in A$ and $j\in B$. In particular, if $\textit{pick}_k[i_{A, B}]=\bot$ we let $j'=j=i_{A, B}$.

    The event $E_{A, B}$ requires that, if $j'\neq j$, then on the last orientation $\vec{G}_k$ there must be a directed edge from $j'$ to $j$. (That is, swap the role of $A$ and $B$ when needed.)
    \item Let $P'_{A, B}$ be the contracted path after the last round. Note that $j\in P'_{A, B}$ and $j\neq c_A$.
    Let $j''$ be the neighbor of $j$ on $P'_{A, B}$ toward $c_A$.
    Let us define the term \emph{forward trajectory} on the union of oriented graphs $G^\star := \vec{G}_0 \cup \vec{G}_1 \cup \cdots \cup \vec{G}_{k-1}$, where each directed edge is associated with the round index. A forward trajectory is then a path $P^\star$ on $G^\star$ such that the round indices is non-decreasing.

    The event $E_{A, B}$ requires
    the existence of a \emph{forward trajectory} $P_A^\star$ from $c_A$ to $j''$ with $V(P_A^\star)\subseteq V(P_{A, B})$.
\end{enumerate}
If $E_{A, B}$ holds, the algorithm adds the tuple $(A, B, j)$ to the output list $\mathcal{O}$.

\paragraph{Analysis.} We verify each item in the statement of \Cref{lemma:randomized-charging-scheme} as follows.

\begin{lemma}\label{lem:chk1}
    If $(A, B, j)\in\mathcal{O}$, then $j\in B$.
\end{lemma}
\begin{proof}
This is guaranteed by Item 4 of the event $E_{A, B}$.
\end{proof}

\begin{lemma}\label{lem:chk2}
    For each district $A\in \mathcal{S}$ the set $J_A$ contains at most $d^{k+1}{\binom{2k}{k}}$ distinct vertices.
\end{lemma}
\begin{proof}
Each forward trajectory can be uniquely encoded with a sequence of ``next round'' or ``go to the $i$-th neighbor'', where the ``next round'' appears exactly $k$ times and the ``go to'' appears for at most $k$ times (this is because the radius of $A$ is at most $k$).
So there are at most $d^{k}\binom{2k}{k}$ vertices of $j' \in A$.

By Item 4 of the event $E_{A, B}$, there is a directed edge going from $j'$ to $j$ in the last oriented graph $\vec{G}_k$. This bounds the size of $J_A$ to be at most $d^{k+1}\binom{2k}{k}$.
\end{proof}

\begin{lemma}\label{lem:chk3}
    The probability of an intersecting pair $(A, B)$ appearing in the output list (in either the form $(A, B, j)$ or $(B, A, j')$) is at least $1/(d+1)^{k^2}$.
\end{lemma}

\begin{proof}
For any $r$, $0\le r\le k$, let $P^{(r)}_{A, B}$ be the contracted path of $P_{A, B}$ in $G_r$.
Let $\vec{P}^{(r)}_{A, B}$ be the associated sequence of directed edges on $\vec{G}_r$ after the orientation.

It suffices to observe that, for rounds $r=1,2, \ldots, k-1$, if the oriented sequence of edges $P^{(r-1)}_{A, B}$ do not belong to one of the following 9 cases, with probability at least $1/(d+1)^k$, the path between $c_A$ and $i_{A, B}$ and the path between $i_{A, B}$ and $c_B$ will become shorter.
These 9 cases are:
\begin{itemize}[itemsep=0pt]
    \item $c_A \rightarrow u \leftarrow i_{A, B} \rightarrow v \leftarrow c_B$
    \item $c_A \rightarrow u \leftarrow i_{A, B} \leftrightarrow  c_B$ (2 cases)
    \item $c_A \leftrightarrow i_{A, B} \rightarrow v \leftarrow c_B$ (2 cases)
    \item $c_A \leftrightarrow i_{A, B} \leftrightarrow c_B$ (4 cases)
\end{itemize} 
By the observation, after $k-1$ rounds, the oriented sequence of $P^{(k-1)}_{A, B}$ on $\vec{G}_{k-1}$ must be in one of the 9 cases stated above (we call them \emph{directed-minimal} paths).
Now, consider the last round $k$. For each of the 9 cases, there is at least one scenario (i.e., with probability at least $1/(d+1)^5$) such that $E_{A, B}$ or $E_{B, A}$ holds, which completes the proof.
\end{proof}


\begin{proof}[Proof of \Cref{lemma:randomized-charging-scheme}]
It follows from \Cref{lem:chk1}, \Cref{lem:chk2}, and \Cref{lem:chk3}.
\end{proof}

\subsection{Bounded Expansion Graphs}\label{subsec:expansion}

While the correlation ratio is constant on $H$-minor-free graphs (\Cref{thm:main}), it has an $\Omega(\sqrt{n})$ lower bound on general graphs~\cite{Dharangutte2025-ry}. This dichotomy naturally raises the question of how broadly the constant bound persists, and whether one can obtain a smooth tradeoff between these two regimes. In this subsection, we take a step in this direction by considering bounded expansion graphs~\cite{Nesetril2008-af,Nesetril2008-vp,Nesetril2008-lt} (refer \Cref{def:bounded_exp_graph}), a class that formalizes the notion of \emph{everywhere sparsity}.


To better demonstrate our results, we first introduce a recursively defined sequence $\{d^f_r\}_{r=0}^\infty$.

\begin{definition}\label{def:dr}
    Let $f: \mathbb{N}\to \mathbb{R}_{\ge 1}$ be any increasing function.
    Define $d^f_0 := f(0)$.
    For each integer $r\ge 1$, define 
    \[
    d^f_r := f\left(2^r\cdot d^f_0d^f_1\cdots d^f_{r-1}\right).
    \]
\end{definition}
If the context is clear we will omit $f$ and use $d_r$ to denote $d^f_r$.
We remark that when $f$ is a polynomial, $d^f_r$ is doubly exponential in $r$.

\begin{theorem}\label{thm:bdd-expansion}
Let $G$ be an $f$-bounded expansion graph. Let $\mathcal{S}$ be a collection of districts with a strong radius bound $k$. Let $\{x_S\}$ be any feasible fractional solution for packing $\mathcal{S}$.
Then, we have:
\[
\sum_{A\cap B\neq\emptyset} x_Ax_B \le \left(4^k d_k^{2k^2 + k}\right)\cdot{\textstyle\binom{2k}{k}} \cdot \sum_{S\in\mathcal{S}} x_S\ .
\]
\end{theorem}
It is straightforward to check that (from \Cref{def:dr}), whenever $k$ is a constant, despite being enormous, the approximation ratio is still  a constant.

Algorithm $\mathcal{R}$ in~\cref{lemma:randomized-charging-scheme} does not directly imply \Cref{thm:bdd-expansion}, because in each contraction which produces $G_r$ from $G_{r-1}$, the diameter of each contracted vertex subset is not bounded.
Fortunately, with a slight tweak to algorithm $\mathcal{R}$, we are able to guarantee the diameter on each contraction, thereby proving \Cref{thm:bdd-expansion}.

\begin{proof}[Proof of \Cref{thm:bdd-expansion}]

We first describe the modified algorithm $\mathcal{R}'$, and then provide the modified analysis that is analogous to \Cref{lem:chk2} and \Cref{lem:chk3}.

\medskip 
\noindent\textbf{The Modified Algorithm $\mathcal{R}'$.} 
In each round $r$ of contraction, the algorithm first computes an \emph{acyclic} orientation $\vec{G}_{r-1}$ via the greedy ordering algorithm\footnote{The algorithm repeatedly identifies a vertex with the minimum current degree, orients all incident edges as outgoing edges, and then remove this vertex and all incident edges until the graph become empty.}  for computing the degeneracy $d_{r-1}$. This also computes a $(d_{r-1}+1)$-vertex-coloring for $G$ with colors $\{1, 2, \ldots, d_{r-1}+1\}$.
Then, the algorithm computes a random shuffle $\sigma$ to the colors. Each vertex $v$ obtains a color $c(v)\in\{1, 2, \ldots, d_{r-1}+1\}$.

The algorithm then assigns the value $\mathit{pick}_r[v]$ for each vertex $v$ with a slightly modified distribution: for each neighbor $u$ of $v$, if $c(u) < c(v)$, assigns $\mathit{pick}_r[v] = u$ 
with probability $1/(d_{r-1}+1)$. Otherwise, set $\mathit{pick}_r[v]=\bot$.

The rest of the algorithm is the same as algorithm $\mathcal{R}$.

\medskip 
\noindent\textbf{The Analysis.}
To complete the proof of \Cref{thm:bdd-expansion}, it suffices to show two statements, analogous to \Cref{lem:chk2} and \Cref{lem:chk3}.
First, we analyze contracted components' diameter. 
Consider any round $r$ of contraction and any vertex $v$. The sequence formed by successively applying $\mathit{pick}_r[\cdot]$ has length at most $2d_{r-1}$ on $G_{r-1}$. Therefore, by induction, the diameter of each contracted component of $G$ at round $r$ for the shallow minor is at most $f\left( 2^r d_0d_1\cdots d_{r-1}\right)=d_r$.

Let $(A, B)$ be any intersecting pair and $P_{A, B}$ be a connecting path from $c_A$ to $c_B$.
To count the number of forward trajectories, we simply observe that the outdegrees for each vertex in any contracted graphs after all $k$ rounds is $d_k$. Any forward trajectories takes at most $k$ steps hopping forward and $k$ steps advancing to the next round. Thus the size of $J_A$ in the output list can be bounded by $d_k^k\binom{2k}{k}$ as desired.

Suppose the contracted path $P_{A, B}^{(r)}$ at round $r$ has length at least 5. So none of the nine directed-minimal cases occur.
In this case, one of the partial path between $c_A$ and $i_{A, B}$ or between $c_B$ and $i_{A, B}$ must have at least length three.
There must be an edge $x\to y$ in the oriented graph that does not invalidate the event $E_{A, B}$ (or $E_{B, A}$) whenever $\mathit{pick}_r[x]=y$.
Now, the probability that the algorithm assigns $\mathit{pick}_r[x]=y$ is no longer $1/(d_{r-1}+1)$. 
In fact, since shuffling colors on $\vec{G}_{r-1}$ is independent to the initial assignment to $\mathit{pick}_r[x]$, we know that $\Pr(c(y) < c(x))  = 1/2$.
Hence, we conclude that with probability $(A, B)$ appearing in the output list (in either form $(A, B, j)$ or $(B, A, j')$) is at least
\[
\prod_{r=1}^k \left(\frac{1}{2(d_{r-1}+1)^k}\right)^2 = 4^{-k} \cdot \prod_{r=1}^{k} \left(d_{r-1}+1 \right)^{-2k} \ge 4^{-k}d_k^{-2k^2},
\]
and the statement follows.
\end{proof}

\paragraph{Remark.} Although \Cref{thm:bdd-expansion} applies to strong radius districts, it does not immediately yield an efficient algorithm for our packing problem. The main obstacle lies in algorithmically extracting a polynomial-sized support from the linear program. In particular, employing either the multiplicative weight update method or the ellipsoid method necessitates an efficient implementation of the separation oracle. However, designing such an oracle appears to be surprisingly challenging, and we currently do not have a solution, even for everywhere sparse graphs.

On the other hand, when every vertex neighborhood has bounded treewidth (as is the case for \emph{apex-minor-free graphs}) an efficient implementation of the separation oracle becomes possible.
In the next subsection, we present such a separation oracle.

\subsection{Separation Oracle for Districts of Bounded Radius }\label{sec:separation}
Here, we outline how to implement separation oracles to search for a weakly violating connected subgraph. The main idea is to formulate the problem as a \emph{vector-valued connected subgraph sum} problem where each vertex has a vector that captures its objective weight, features, and dual variables.

Given an undirected graph $G = (V,E)$ where each node $v\in V$ has a vector-valued weight $\vw(v)\in \mathbb{Z}_{\ge 0}^d$, the \emph{Connected Subgraph Sum} problem defines the set $W^*$ of all possible total weights achievable by a connected and compact subgraphs of $G$:
\begin{equation}\label{eq:subgraphsum}
    W^* = \left\{\vw(T) = \sum_{v\in T} \vw(v): \text{the induced subgraph } G[T] \text{ is connected and compact}\right\}.
\end{equation}
We define the set of feasible weight vectors $W^*$ only based on the topological constraints (connectivity and compactness). The composition constraints (e.g., balancedness or thresholds) are encoded in the vector values and handled subsequently via the linear score function $\ell$ and the domination idea.

As the set $W^*$ can be exponentially large, explicitly computing and verifying composition constraints in~\cref{eq:subgraphsum} is infeasible.  Instead, we propose the \emph{$(\ell,\epsilon)$-Trimmed Connected Subgraph Sum} problem. This seeks a succinct subset that approximates every feasible weight vector multiplicatively, while maximizing a specific linear score $\ell$, which encodes the necessary information to satisfy the composition requirements.

\begin{mdframed}
\textsc{$(\ell,\epsilon)$-Trimmed Connected Subgraph Sum}\\
\textbf{Input:} A graph $G = (V, E)$ with a vector-valued weight function $\vw:V\to \mathbb{Z}_{\ge 0}^d$, an error parameter $\epsilon\ge 0$, and a linear function $\ell:\mathbb{Z}_{\ge 0}^d\to \mathbb{Z}$.\\
\textbf{Goal:}
Compute a subset $W\subseteq \mathbb{Z}_{\ge 0}^d$ such that $W\subseteq W^*$ and $W$ is an $(\ell, \epsilon)$-trimming of $W^*$. That is, for every $\vz = (z_1,\dots,z_d)\in W^*$, there exists a representative $\vz' = (z_1',\dots,z_d')\in W$ satisfying:
\begin{align}
e^{-\epsilon} &\le \frac{z_i'}{z_i}\le e^\epsilon \quad \text{for all } i = 1,\dots,d, \text{ and}\label{eq:subgraphsum1}\\
\ell(\vz') &\ge \ell(\vz).\label{eq:subgraphsum2}
\end{align}
\end{mdframed}

We say that $\vz'$ \emph{$\epsilon$-approximates} $\vz$ if condition (\ref{eq:subgraphsum1}) holds, and that $\vz'$ \emph{$\ell$-dominates} $\vz$ if condition (\ref{eq:subgraphsum2}) holds.

Crucially, this problem is computationally tractable on graphs of bounded treewidth. The main result of this section is an FPTAS for the trimmed connected subgraph sum problem, which serves as the tool for our separation oracles.
\begin{lemma}\label{lem:allsubsetsum}
    Given $G = (V,E)$, a radius bound $k$, vector-valued weights $\vw:V\to \gridnd$, an error parameter $\epsilon> 0$, and a linear function $\ell$, there is an algorithm for $(\ell, \epsilon)$-trimmed connected subgraph sum with strong radius bound that has running time 
  $$O\left( n^{2+2d} \epsilon^{-2d} (\radiusbd+1)^{2(tw+1)} 2^{O(tw \log tw)} \left(\ln R\right)^{2d} \right),$$
  where $n$ is the number of vertices and $tw$ is the treewidth of $G$, and $R = \max_i\sum_{v\in V} w_i(v)$.
\end{lemma}

Armed with this efficient algorithm, we now demonstrate how the separation oracles for our districting problems can be reduced to instances of the trimmed connected subgraph sum problem.

\begin{lemma}[Separation Oracle for $c$-Balanced Districts]\label{lem:separation_balanced2subgraphsum}
Consider the $c$-balanced districting problem where each vertex $v$ has an objective weight $w(v)$ and feature weights $w_1(v), w_2(v)$. The separation oracle seeks a district $S_{\max}$ maximizing objective weight subject to the balance condition and a violated dual constraint so that 
$$S_{\max} \in \arg\max\left\{w(S'): y(S') < (1-\epsilon)w(S') \text{ and $S'$ is $c$-balanced} \right\}.$$

This can be reduced to solving two instances of the trimmed connected subgraph sum problem. We define the vector-valued weight as $\vw = (w_1, w_2, w, y)$  and use the linear functions:
$$ \ell_1(\vz) = (c-1)z_1-z_2 \quad \text{and} \quad \ell_2(\vz) = (c-1)z_2-z_1 $$By finding an $(\ell_1, \epsilon)$-trimmed and $(\ell_2, \epsilon)$-trimmed subset, we can identify a weakly violating district $S$ that satisfies $y(S) < (1-\frac{\epsilon}{2}) w(S)$, approximates the optimal weight $w(S) \ge \frac{1}{2} w(S_{\max})$, and satisfies the balance condition~\cref{eq:def_balanced}.
\end{lemma}

\begin{lemma}[Separation Oracle for Threshold Districts]\label{lem:separation_threshold2subgraphsum}
Consider the $B$-threshold districting problem with objective weight $w(v)$ and feature weight $w_1(v)$ subject to $w_1(S) \ge B$. The separation oracle additionally has a dual variable $y$ and seeks a district maximizing objective weight subject to the threshold and dual constraints.  

We can reduce this to solving the trimmed connected subgraph sum problem with vector weights $\vw = (w_1, w, y)$  and the linear function $\ell_{th}(\vz) = z_1$. The resulting set allows us to compute a weakly violating district satisfying $y(S) < (1-\frac{1}{2}\epsilon) w(S)$, approximate objective weight $w(S) \ge \frac{1}{2} w(S_{\max})$, and the threshold constraint~\cref{eq:def:threshold}.
\end{lemma}

\begin{example}[Connected Knapsack Problem~\cite{Dey2024-jg}]  Given a graph $G$, where each vertex has a value $\alpha(v)$ and a cost $\beta(v)$, and a global budget $B$, the goal is to find a connected subgraph $U$ maximizing $\alpha(U)$ subject to $\beta(U) \le B$. This is exactly recovered by setting $\vw(v) = (\alpha(v), \beta(v))$ and solving for the appropriate approximation with $\ell(\alpha, \beta) = -\beta$.
\end{example}

\subsection{Proofs to \Cref{thm:k-apex-minor} and \Cref{thm:star-bounded-expansion}}

Finally, we conclude this section by describing how to apply separation oracles provided in the previous section. In particular, the new separation oracle leads to $O(1)$-approximation algorithms for bounded strong radius districts on apex-minor-free graphs.

\begin{proof}[Proof of \Cref{thm:k-apex-minor}]
To apply \Cref{lem:allsubsetsum} (with the reduction adapted from \Cref{lem:separation_balanced2subgraphsum} on $d=4$ and \Cref{lem:separation_threshold2subgraphsum} on $d=3$) in the separation oracle, the algorithm enumerates all vertices $v \in V$ as centers and considers the subgraph induced by vertices within distance $k$ of $v$.
As $G$ is apex-minor-free, each such subgraph has bounded treewidth, which makes \Cref{lem:allsubsetsum} applicable, and therefore the theorem follows.
\end{proof}

\begin{proof}[Proof of \Cref{thm:star-bounded-expansion}]
We adopt the implementation of the separation oracle from~\cite{Dharangutte2025-ry}, enumerating each vertex as a star center and invoking an approximate knapsack algorithm. The theorem follows from combining the analysis of \Cref{subsec:expansion}.
\end{proof}

\section{PTAS for $\delta$-Relaxed Districting on Apex-Minor-Free Graphs}\label{sec:minorfree}


In this section, we show that allowing relaxation to the composition requirements helps us obtain PTAS to the districting problem on bounded treewidth graphs. Combining this with the standard Baker's approach~\cite{Baker1994-ih}, we obtain PTAS for $\delta$-relaxed districting on planar and apex minor-free graphs.

\begin{theorem}
\label{thm:ptas_relax_aminor_free}
    There exists a polynomial-time algorithm that computes an $(1+\eps)$-approximate solution for packing $\delta$-relaxed $c$-balanced or threshold districts with bounded strong or weak radius on planar graphs and apex-minor-free graphs.
\end{theorem}


We first state the result that obtains a $(1+\eps)$-approximate solution for packing $\delta$-relaxed districts on bounded treewidth graphs.  The proof employs dynamic programming, combining \cref{lem:allsubsetsum} with the techniques of \cite{Dharangutte2025-ry}.  We defer the proof for \Cref{cor:ptas_relax_bounded} to~\cref{app:multisubsetsum}.
\begin{lemma}\label{cor:ptas_relax_bounded}    There exists a polynomial-time algorithm that computes an $(1+\eps)$-approximate solution for packing $\delta$-relaxed $c$-balanced or $B$-threshold districts with bounded strong or weak radius on graphs of bounded treewidth.
\end{lemma}


Using \Cref{cor:ptas_relax_bounded}  as a black box, we prove the main result for $\delta$-relaxed districting.


\begin{proof}[Proof of \Cref{thm:ptas_relax_aminor_free}]

The algorithm is as follows. Start with an arbitrary vertex $r$ as the root and construct a BFS tree and let $L_j$ be the $j^{\text{th}}$ level which are vertices at hop distance $j$ from $r$. Now, for an integer $t > 0$, and  shift $i \in \{0, \cdots, t-1\}$, remove vertices in level $i \mod{t}$, and let $G^{(i)}$ be the induced graph on remaining vertices. 
Notice that each connected component in $G^{(i)}$ is fully contained in some $(t-1)$ consecutive levels of the BFS tree. 
For planar or apex-minor-free $G$, note that each connected component in $G^{(i)}$ now has bounded treewidth \cite{Eppstein2000-ry} (specifically, $O(t)$).
Use the algorithm from \Cref{cor:ptas_relax_bounded} with parameter $\varepsilon/2$ and solve the problem on each connected component separately and let the union of these solutions for a fixed $i$ be $\mathcal{T}_i$. Since districts cannot use vertices from different levels, taking union gives us a valid districting. Return the solution with largest weight $w(\mathcal{T}_i)$ across all $i$.

Now, we show that setting $t = O(k/\varepsilon)$ suffices. Let $\mathcal{T}^*$ be the optimal solution on $G$ with $w(\mathcal{T}^*) = \text{OPT}$. A radius $k$ district in $\mathcal{T}^*$ is valid in some connected component in $G^{(i)}$ if no vertex in the district (or the path considered by the optimal solution for weak radius setting) lies in the deleted levels. Hence, for at most $2k + 1$ different shifts ($i$ values), this particular district (or the path from the center vertex to some vertex in the district) is not valid. Now, for shift $i$, let $S_i$ be the set of districts from $\mathcal{T}^*$ that intersect with some deleted layer. The optimal districting in $G^{(i)}$ has weight $\geq \text{OPT} - w(S_i)$. Now consider the total weight lost across different shifts. We have

$$\sum_{i=0}^{t-1} w(S_i) \leq (2k+1) \sum_{T \in \mathcal{T^*}} w(T) \leq (2k+1) \text{OPT}.$$ 

By pigeonhole principle, for $t=2(2k+1)/\varepsilon$, there is some $i^*$ such that 
\[
w(S_{i^*}) \leq \frac{2k+1}{t} \text{OPT} \leq \frac{\varepsilon}{2} \text{OPT}
\]

As a result, for this particular shift $i^*$, the optimal districting in $G^{(i^*)}$ has weight at least $\geq \text{OPT} - w(S_{i*}) \geq \left(1 - \frac{\eps}{2}\right)\text{OPT}$.
Combining this with the $(1-\varepsilon/2)$ approximation on $G^{(i^*)}$ from \Cref{cor:ptas_relax_bounded}, we obtain that the returned solution has weight at least $\left(1-\frac{\eps}{2}\right)^2 \text{OPT} \geq (1-\varepsilon) \text{OPT}$. 
\end{proof}

\section{Hardness Results}
\label{sec:hardness}
In this section we prove our hardness results. Note that the hardness for maximum packing of $c$-balanced districts were shown in \cite{Dharangutte2025-ry}. Here, we show the hardness of packing under the weight threshold condition. 


\subsection{Hardness of Packing Compact Subgraphs with Weight Threshold}


We first start by showing the problem is NP-hard even when restricting input instances to trees.

\begin{theorem}
    \label{thm:hardness-tree}
    The $B$-threshold districting problem with (weak or strong) radius-$k$ bounded districts is NP-hard when $G$ is a tree.
\end{theorem}
\begin{proof}
    We first focus on star districts and show a reduction from the Knapsack Problem. Given a knapsack problem with $n$ items, where item $i$ has weight $w_i$ and utility $u_i$, we want to select items such that utility is at least $U$, and minimize the total weight of the selected one. This is the dual version of the standard knapsack problem which fixes the cap of the total weight and maximizes utility, and is still NP-hard. Without loss of generality, we can assume that $u_i>1$ and $w_i<1/n$. This can be achieved by scaling up all utility values and the utility goal $U$ and scale down the weights. It does not change the problem structure. 

Now we describe the instance for the districting problem. Consider a tree rooted at $v$ with children of $v$ as $v_1, v_2, \dots, v_n$ corresponding to the items. 
Let $w(v) =1$, and $B=2U$. 
Node $v_i$ has two children $x_i$ and $y_i$.  Define 
$w(x_i)=B$, 
$w(y_i)=w_i$,
$w(v_i)=u_i$.
Further, we add a separate neighbor/child $v_0$ of $v$ with $w(v_0)=U-1$. The tree structure is detailed in \Cref{fig:knapsack-tree}.

    \begin{figure}[htbp]
    \centering
    \includegraphics[width=5cm]{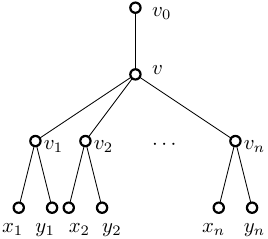}
    \caption{Instance for districting on the tree from a given knapsack instance.}
    \label{fig:knapsack-tree}
    \end{figure}

Now we consider what are the possible candidate districts that have a total weight at least $B$. First we can define a district centered at $v$, including $v_0$ and some set of vertices in $\{v_1, \cdots, v_n\}$. The population of $v$ and $v_0$ is $U=B/2$, so we must include additional vertices from $\{v_i\}$ to gain additional population of at least $U$. On the other hand, we can also define a district centered at $v_i$, which includes $x_i, y_i$ (and possibly $v$). However, we cannot have both a district centered at $v$ including $v_i$ and a district centered at $v_i$, as we want disjoint districts. If $v_i$ is included in the district at $v$, then we can still define a district with a solo vertex $x_i$ (of population $B$), but $y_i$ cannot be included in any district.

Thus the main decision to make is whether we can include a district centered at $v$ or not. To do so, we select a subset of vertices  in $\{v_1, \cdots, v_n\}$ which has a total weight at least $U$. For these ones selected, the corresponding vertices $y_i$ will not be included anywhere -- causing a loss of $w_i$ from the potential districts centered at $v_i$. However, since $\sum_i w_i<1$, it would always be more beneficial to create a district at $v$ whenever possible. Specifically, the total coverage of the two cases is
\begin{itemize}
\item If we do not have a district in $v$, the optimal solution is $1+ nB+\sum_i u_i +\sum_i w_i$.
In the above, the $1$ comes from some district at $v_i$ that can include $v$.

\item If we do include a district in $v$, which includes $v_0$ and a subset of vertices $S\subseteq \{v_1, \cdots, v_n\}$, then the solution is  $B/2+ nB+\sum_i u_i +\sum_i w_i - \sum_{i\in S} w_i$.
\end{itemize}

In other words, the best solution is to try to include a district at $v$ and the set $S$ of children included identifies the items that have the smallest total weight yet have utility to be at least $U$. Notice that this is precisely the optimal solution for the knapsack problem. Next, we note that the same reduction extends to radius-$k$ districts, for any fixed $k \ge 1$, by replacing every edge of the above tree by a path of length exactly $k$, and assigning weight $0$ to all newly introduced subdivision vertices. The resulting graph is still a tree, and all original districts retain the same weights. The same proof goes through by noting that on a tree, the weak-radius and strong-radius notions are equivalent for connected districts.

\end{proof}

Next, we show hardness for general graphs.

\begin{theorem}
\label{thm:hardness-max-packing}
The $B$-threshold districting problem with (weak or strong) radius-$k$ bounded districts is NP-hard when $G$ is a bounded degree graph or a planar graph. Moreover, on general graphs, it is NP-hard to approximate better than $O(n^{1/2-\delta})$ for any $\delta>0$. The results hold even when each subgraph is required to be a star (radius $1$).
\end{theorem}
\begin{proof}
We show hardness and hardness of approximation for $k=1$ by using the same reduction from maximum independent set as in~\cite{Dharangutte2025-ry}. Specifically, given a graph $G=(V, E)$, the maximum independent set problem finds a subset of vertices $S\subseteq V$ with maximum cardinality such that no vertices in $S$ have edges between them. We now construct a graph $G'$ that starts with $G$. For each edge $(u, v)$ in $G$, split it into two edges in $G'$ -- create a new vertex $p_{uv}$ which connects to $u$ and $v$ and remove the edge $(u, v)$. In addition, suppose the maximum degree in $G$ is $\Delta$. For a vertex $v$ in $G'$ (corresponding to a vertex in $G$) whose degree $d(v)$ is less than $\Delta$, create $\Delta-d(v)$ additional neighbors (called filler nodes). We then set $w(v)=1$ for each vertex in $G'$, $B=\Delta+1$ and $k=1$. Without loss of generality, we can assume that $\Delta\geq 3$ and thus $B\geq 4$. This means that any district with weight at least $B$ in $G'$ must be at a vertex which corresponds to vertices in $G$. Further, this district will need to use all $\Delta$ neighbors in order to have enough total weight. This means that any two districts in $G'$ cannot have their centers to be adjacent in $G$. Thus any independent set $S$ in $G$ can be translated to a collection star districts in $G'$, with each district centered on a vertex in $S$ with all its $\Delta$ neighbors. Similarly, any districting solution will produce a set of vertices in $V$ that are independent.  
The total weight of the optimal districting problem of $G'$ is $(\Delta+1)|S|$ where $S$ is the maximum independent set in $G$. Similar to the argument in~\cite{Dharangutte2025-ry}, since maximum independent set cannot be approximated by a factor of $n^{1-\delta}$ for any $\delta>0$,  this means that our districting problem cannot be approximated better than $n^{1/2-\delta}$ in a general graph for any $\delta>0$, unless $\mathsf{P}=\mathsf{NP}$. Further, it is $\mathsf{NP}$-hard to solve the districting problem for graphs of bounded degree or on planar graphs, since the maximum independent set for these graphs remains $\mathsf{NP}$-hard. 

To extend the hardness for radius-$k$ districts, we simply replace the edges we add in constructing $G'$ with paths of length $k$. So, instead of connecting $u$ and $p_{uv}$ with a single edge, we connect them with a path of length $k$ (i.e., with $k-1$ new vertices on the path from $u$ to $p_{uv}$). We do the same for $v$ and $p_{uv}$ and for each vertex $v$ with their respective $\Delta - d(v)$ filler neighbors. We set $B = 1 + \radiusbd\Delta$. A valid district must be centered at a vertex $v$ corresponding to vertices in the original graph $G$ and include $\Delta$ paths from $v$, each of length $k$. Following the same arguments completes the proof.
\end{proof}

\begin{remark}
    \label{remark:relaxed}
    The reduction in \Cref{thm:hardness-max-packing} also works for the relaxed version of $B$-threshold districting. This is because for fixed $\delta$ and radius-$\radiusbd$ districts, we can always choose $B$ such that $\Delta < (1-\delta)B \leq \Delta + 1$, which forces the $B$-threshold districts to be the exact same as described in the proof. 
\end{remark}

The above results cover the hardness for $B$-threshold districting problem. Now, considering the $c$-balanced districting, a natural direction for generelization is considering three colors for balanced condition. We first start by defining the problem and then prove our hardness result for this version.

\begin{definition}[$c$-balanced district with three weight types]
Let each vertex $v \in V$ have three non-negative weights $w_1(v), w_2(v), w_3(v)$, and define the objective weight of a vertex as $w(v) = w_1(v) + w_2(v) + w_3(v)$. For a subset $S \subseteq V$, let $w_i(S) = \sum_{v \in S} w_i(v)$ for $i \in \{1,2,3\}$, and define $w(S) = \sum_{v \in S} w(v)$.

We say that $S$ is \emph{$c$-balanced} (with respect to three weight types) if
\[
\min\{w_1(S), w_2(S), w_3(S)\} \geq \frac{w(S)}{c}.
\]
\end{definition}

\begin{theorem}
\label{thm:three-colors}
The $c$-balanced districting problem with three weight types is NP-hard to approximate within any multiplicative factor.
\end{theorem}

\begin{proof}
We construct an instance where each vertex has three weight components $w_1(v), w_2(v), w_3(v)$, and districts must be $c$-balanced.

The instance includes:
\begin{itemize}
    \item One special vertex with $w_1 = \frac{c-2}{c}$, $w_2 = \frac{1}{c}$, and $w_3 = 0$.
    \item A large number of other vertices, each with $w_1 = w_2 = 0$ and various positive values of $w_3$.
\end{itemize}

Any valid $c$-balanced district must include the special vertex and a subset of $w_3$-only vertices such that the total weight satisfies:
\[
\min\left\{ w_1(S), w_2(S), w_3(S) \right\} \geq \frac{w(S)}{c}.
\]

Since $w_1(S) + w_2(S) = \frac{c-1}{c}$ is fixed, and all additional weight comes from $w_3$, the only possible value of $w_3(S)$ that satisfies the balance condition is exactly $\frac{1}{c}$. Any smaller value causes $w_3$ to fall below the threshold; any larger value causes $w_2$ to fall below.

Thus, finding even a single valid district requires finding a subset of vertices whose $w_3$-weights sum exactly to $1/c$, which is equivalent to solving an exact subset sum problem, which is known to be NP-complete. Moreover, since any valid solution must be exact, even approximating the maximum total weight of valid districts within any factor would allow us to solve this exact subset sum instance.
\end{proof}

Finally, we show that if we consider compactness requriment to be diameter bounded subgraphs, then finding even a single district becomes NP-hard.


\begin{theorem}
\label{thm:hardness-diam-single}
    Given a graph $G$ with integer weights $w: V \rightarrow \mathbb{Z}_{\geq 0}$ on vertices, integers $B > 0$ and $k \geq 1$ the problem of finding a vertex subset $S$ such that $G[S]$ is connected, diam$(G[S]) \leq k$ 
    and $w(S) \geq B$ (a $B$-threshold district) is NP-hard. Moreover, the problem remains NP-hard if instead we require the weak diameter of $G[S] \leq k$.
\end{theorem}

\begin{proof}
    Note that for $k = 1$, the problem reduces to finding a clique (which is hard) so we focus on $k \geq 2$. We first consider the case of strong diameter and fixed $k=2$ and show reduction from decision version of clique problem. Recall for the clique problem, given a graph $G$ and integer $T$ deciding if there is a clique of size $\geq T$ is NP-hard. 
    
    Given $G = (V,E)$ and $T$, we construct the instance $G' = (V',E',w)$ for the weight threshold problem as follows. For each $v \in V$, create a corresponding vertex $v' \in V'$ with $w_{v'} = 1$. For each edge $(u,v) \in E$, create a vertex $x_{(u,v)}$ with $w_{x_{(u,v)}} = 0$ and add edges $(u',x_{(u,v)})$ and $(v',x_{(u,v)})$. Add edges between all $x_{(u,v)}$ vertices in $G'$ (make it a clique).  Set $B = T$, the weight threshold for district. Observe that for any $(u,v) \in E$, distance between $u'$ and $v'$ in $G'$ is 2 (through $x_{u,v}$) and for any non-edge $(u,w)$ in $G$, distance between $u'$ and $w'$ is at least 3 (distance is 3 iff $u$, $w$ has at least one incident edge in $G$). 
    Now, for a clique $S$ in $G$, $S' = \{ v' | v \in S\} \cup \{x_{(u,v)} | u \in S \;\text{and}\; v\in S\}$ forms a $B$-threshold district for $k=2$ with $w(S') = |S|$. 
    This is because, $G'[S']$ is connected, pairwise distance of all vertices corresponding to clique vertices is at most 2, distance of $u'$ and $x_{(u,v)}$ for $u$ and $v$ in $S$ is at most 2 and all $x_{(u,v)}$ vertices have edges between them. Now for the other direction let $S'$ be any $B$-threshold district for $k=2$ and $B = T$ in $G'$. Ignoring all the $x_{(u,v)}$ vertices in $S'$, the remaining vertices correspond to a clique in $G$ of size at least $T$. This is again because for vertices in $G'$ corresponding to vertices in $G$, distance is at most 2 iff there is an edge connecting them in $G$, and only these vertices have non-zero weight in $G'$. This completes the proof for $k=2$.


For $k \ge 3$, we modify the construction of $G'$ as follows. 
Fix integers
\[
\ell_1 = \left\lfloor \frac{k}{2} \right\rfloor 
\qquad \text{and} \qquad 
\ell_2 = k+1 - 2\ell_1.
\]
Observe that $2\ell_1 \le k$ and $2\ell_1 + \ell_2 = k+1$.

For each vertex $v \in V(G)$, create a vertex $v' \in V(G')$ with weight $w_{v'} = 1$. 
For each edge $(u,v) \in E(G)$, create a vertex $x_{(u,v)}$ with weight $w_{x_{(u,v)}} = 0$. 
Connect $u'$ to $x_{(u,v)}$ by a path of length $\ell_1$, and connect $v'$ to $x_{(u,v)}$ by a path of length $\ell_1$. 
All internal vertices on these paths are new vertices of weight $0$, and all such paths are internally vertex-disjoint.
Furthermore, for every unordered pair of distinct edge-vertices 
$x_{(u,v)}$ and $x_{(u',v')}$, 
connect them by a path of length $\ell_2$, whose internal vertices are new and have weight $0$. 
All these paths are internally vertex-disjoint from one another and from the previously constructed paths. 
Set the weight threshold $B = T$.

We now verify the distance properties. 
If $(u,v) \in E(G)$, then
\[
d_{G'}(u',v') = 2\ell_1 \le k.
\]
If $(u,v) \notin E(G)$, then any path from $u'$ to $v'$ must pass through some edge-vertex adjacent to $u'$ and some edge-vertex adjacent to $v'$, and therefore has length at least
\[
\ell_1 + \ell_2 + \ell_1 = 2\ell_1 + \ell_2 = k+1,
\]
implying $d_{G'}(u',v') > k$.

Consequently, pairs of vertices corresponding to edges of $G$ have distance at most $k$, whereas pairs corresponding to non-edges have distance strictly greater than $k$. The remainder of the proof proceeds as in the $k=2$ case.

    Now, observe that same instance and argument also hold for the case when weak diameter is restricted to be at most $k$. This is because, for a clique in $G$, we have a $B$-threshold district in $G'$ with strong diameter at most $k$, which implies a $B$-threshold district for weak diameter at most $k$. Also, for any connected $B$-threshold district of weak diameter at most $k$ and weight $ \geq T$, the vertices with non-zero weights correspond to a clique of size $\geq T$ in $G$.
\end{proof}

Observe that the hardness of finding just one district implies the hardness for separation oracle (\Cref{eq:violating-district}) on general graphs. This is because separation oracle requires reporting a violating constraint, which is a $B$-threshold district but with small values of sum of dual variables. For $y_v = 0$ for all $v$, the same reduction above gives us the hardness result. 

\section{Conclusion and Open Problems}

In this paper, we investigated several variants of packing compact subgraphs into a larger graph under balancedness/threshold composition constraints. We extend the work of ~\cite{Dharangutte2025-ry}, who focused only star districts, by studying districts of bounded radius. 

Our work provides algorithmic results for two variants of the problem. On one hand, we design an $O(1)$-approximation algorithm under exact $c$-balanced and $B$-threshold districting requirements on apex-minor-free graphs.
This is done by solving a linear program and then applying the rounding framework.
On the other hand, we obtain a PTAS on apex-minor-free graphs by relaxing the constraints, allowing $(1+\delta)c$-balancedness and $(1-\delta)B$ threshold parameters, for any $\delta>0$. This is achieved by dynamic programming on tree decomposition with Baker's method.

Interestingly, we do not immediately see a clean path to get the best of the two worlds --- a PTAS on apex-minor-free graphs with exact $c$-balanced or $B$ threshold constraints.
The rounding framework leaves an $\Omega(1)$ integrality gap, which seems very hard to overcome (as discussed in~\cite{Dharangutte2025-ry}).

On the other hand, for the dynamic programming approach,
The main bottleneck appears to be maintaining intermediate states that encode multiple partial districts with exact composition requirements simultaneously. 
Our algorithm requires ``state blow-up'' and  ``trimming'' procedures at each DP transition step.
The trimming step maintains an invariant that specifically holds for only one partial district in the description of \Cref{eq:subgraphsum2}. Similar to the 3-color-balancedness hardness result (\Cref{thm:three-colors}),
the current trimming algorithm does not naturally extend well to \Cref{lem:trimming} and does not guarantee \Cref{eq:subgraphsum2}. Clearly, if we choose not to trim at all, the size of the state space grows exponentially with $n$, thereby yielding an intractable algorithm.

However, this problem is not as desperate as we see in \Cref{thm:three-colors}. 
One example is that we can obtain a PTAS for \emph{grid graphs} by applying Baker's idea \emph{twice}: we split the grid into $t^2=O(k/\epsilon)\times O(k/\epsilon)$ square subgrids and consider all $t^2$ possible shifts. Since each connected component has only a constant number $\Theta(t^2)=O_{k, \epsilon}(1)$  of vertices, there is a constant-time algorithm for obtaining the optimal solution on each piece, leading to a PTAS overall.  We are therefore optimistic about the following conjecture.

\begin{conjecture}
There exists a PTAS for packing $c$-balanced or $ B$-threshold districts on apex-minor-free graphs.
\end{conjecture}

\paragraph{Efficient Separation Oracles for Various Compact Subgraph Types.}

Our current rounding framework handles exact balancedness constraints only for districts of bounded \emph{strong radius}. In contrast, the relaxed PTAS relies on districts of bounded weak radius. 

\begin{question}
Is the correlation ratio bounded by a constant for bounded weak-radius districts?
\end{question}

\paragraph{Beyond Apex-Minor-Free Graphs.}
Our current techniques rely on dynamic programming over bounded-treewidth graphs. Extending these results to broader graph classes will likely require fundamentally new ideas. In particular, it is natural to ask whether similar approximation guarantees can be achieved for larger graph classes, such as bounded-expansion graphs.

\begin{question}
Does there exist an efficient separation oracle for bounded strong-radius districts on bounded-expansion graphs?
\end{question}

\paragraph{Graphs Parametrized by Worst-Case Correlation Ratio.}
In this paper, as well as in~\cite{ChanH12,Dharangutte2025-ry}, we have seen rounding algorithms whose approximation guarantee is closely related to the correlation ratio.
Roughly speaking, this correlation ratio describes how districts satisfying a certain composition requirement interact with one another in the worst-case vertex-weight assignment.
From a purely graph-theoretical perspective, it gives a flavor of per-graph-class analysis of approximation ratios.
For example, in \Cref{sec:rounding} we have shown that in planar graphs and bounded expansion graphs have a constant worst-case correlation ratio. In~\cite{Dharangutte2025-ry} it shows an $\Omega(\sqrt{n})$ lower bound ratio on general graphs and a $\Omega(1)$ lower bound ratio on planar graphs.
It is not hard to verify (using balanced separators) that on constant treedepth graphs the ratio is $O(1)$ as well and on bounded treewidth and pathwidth graphs the ratio is $O(\log n)$, but we do not know if this polylogarithmic bound is tight.
We are interested in connecting this parameter to other graph parameter and classes in a reversed way, perhaps leading to universal optimal (per-graph topology) analysis on approximation ratios. Finally, 
we wonder whether there are any interesting observations that can be made from the parametrization of this worst-case correlation ratio.

\begin{question}
Are there interesting results when defining the graph classes according to the correlation ratio?
\end{question}

\bibliographystyle{alpha}
\bibliography{districting,cover}
\appendix

\section{Proofs and details of \Cref{sec:separation}}
\subsection{Preliminary on Tree Decomposition}
Given two sets $A$ and $B$, if they are disjoint, we write $A\uplus B$ as their union. 
More generally, for a family of sets $\partition = (P_0,\dots,P_m)$ and $S$, if $\partition$ is a partition of $S$, we write $S = \uplus_i P_i$.

Given an undirected graph $G = (V, E)$, a \emph{tree decomposition} of $G$ consists of a pair $(\mathbb{T}= (V_\mathbb{T}, E_\mathbb{T}), \mathcal{X} = \{X_t\}_{t\in V_\mathbb{T}})$, where $\mathbb{T}$ is a tree, and each node $t\in V_\mathbb{T}$ has a subset $X_t\subseteq V$ called bag such that the following conditions hold.
\begin{itemize}
    \item $\cup_{t\in V_\mathbb{T}} X_t = V$
    \item For all $(u,v)\in E$, there is a $t\in V_\mathbb{T}$ so that both $u, v\in X_t$, and
    \item For any $v\in V$, the subgraph of $\mathbb{T}$ induced by the set $\{t: v\in X_t\}$ is connected.
\end{itemize}
The width of a tree decomposition is $\max_{t\in V_\mathbb{T}}|X_t|-1$.  The treewidth of $G$ is the minimum width over all tree decomposition of $G$, denoted as $tw(G)$.

Moreover, a tree decomposition $\mathbb{T}$ is a \emph{nice} if $\mathbb{T}$ is a binary tree rooted at some node $root$ with $X_{root} = \emptyset$, and each node is of one of the following kinds:
\begin{itemize}
    \item Leaf node $t$ is a leaf of $\mathbb{T}$, and $X_t = \emptyset$.
    \item Introduce vertex node $t$ with $v\in V$ has only one child $s$ and $X_t = X_s\uplus \{v\}$.  We also say that $t$ introduces vertex $v$.
    \item Forget vertex node $t$ with $v\in V$ has one child $t'$ where $X_t\uplus \{v\} = X_{t'}$.
    \item Introduce edge node $t$ with $(u,v)\in E$ has only one child $t'$ and $\{u,v\}\subseteq X_t = X_{t'}$.
    \item Join node $t$ has two children $t_1$ and $t_2$ so that $X_{t_1} = X_{t_2} = X_t$.
\end{itemize}
Given $t\in V_\mathbb{T}$, let $\mathbb{T}_t$ be the subtree of $\mathbb{T}$ rooted at $t$ with node set $V_{\mathbb{T}_t}$.  We denote $G_t = (V_t, E_t)$ as the subgraph of $G$ where the vertex set is $V_t = \cup_{t'\in V_{\mathbb{T}_t}}X_{t'}\subseteq V$. Similarly, for any subset $U\subseteq V$, let $U_t := U\cap V_t$. 

\subsection{Proof of \Cref{lem:allsubsetsum}}\label{app:allsubsetsum}

\begin{proof}[Proof of \Cref{lem:allsubsetsum}]
We use a standard dynamic programming on tree decomposition~\cite{Dey2024-jg} and utilize the shortest path tree for the strong radius constraints.  We maintain a collection of weights in the DP table to approximate solution restricted to the current subtree.  For any $\gamma\in V$, let $\mathcal{U}_\gamma$ be the set of all valid global solutions, i.e., all connected subgraphs $U \subseteq V$ containing $\gamma$ with strong radius at most $\radiusbd$, and 
$W_\gamma^* = \{w(U): U\in \mathcal{U}_\gamma\}.$  Note that $W^* = \cup_\gamma W_\gamma^*$.  
We modify the nice tree decomposition by adding $\gamma$ to all bags.

\paragraph{Dynamic programming states}
We run a dynamic programming algorithm over the tree decomposition. The DP table $D$ is indexed by $(t,s)$ with a node $t\in V_\mathbb{T}$ and a \emph{trace} $s=(\partition, \rho, \pi)$:
\begin{itemize}
\item The current node $t \in V_\mathbb{T}$ in the tree decomposition.
\item A partition $\partition = (P_0,\dots,P_L)$ of the bag $X_t = \cup_{l = 0}^L P_l$. Here $P_0$ is the set of unselected vertices and selected $P_l$ with $l > 0$ correspond to the intersection of the bag with connected components of the partial solution below $t$.
\item A distance profile $\rho: X_t\setminus P_0 \to \{0, \dots, \radiusbd\}$. $\rho(u)$ represents the guessed shortest-path distance from the center $\gamma$ to selected vertex $u$ in a feasible solution $G[U]$.
\item A parent status $\pi: X_t\setminus P_0 \to \{0, 1\}$. $\pi(u)=1$ indicates that $u$ has found a parent $v$ in the partial solution such that $\rho(u) = \rho(v) + 1$. This helps to verify the guessed distance profile.
\end{itemize}

For any $U \in \mathcal{U}_\gamma$, the \emph{trace} of $U$ on the bag $X_t$ is $s = (\partition, \rho, \pi)$ so that 1) $\partition$ is the partition of $X_t$ induced by $U_t = U \cap V_t$ so that the connected components $C_1, \dots, C_L$ of $G_t[U_t]$ satisfy $P_l = C_l \cap X_t$ for all $l$.  2) $\rho$ and $\pi$ induce a valid shortest path tree on $G[U]$: $\rho(u) = d_{G[U]}(\gamma, u)$ for all $u \in X_t$ which is always bounded by $\radiusbd$ and $\pi(u) = 1$ if and only if $u$ has a neighbor $v$ in $U_t$ such that $d_{G[U]}(\gamma, u) = d_{G[U]}(\gamma, v)+1$.  Conversely, $U$ is \emph{compatible} with $s$ if the trace of $U$ is $s$.

We define the feasible weight set $W^*_\gamma(t, s)$ as the collection of weights of the restrictions of all valid global solutions that match this trace:
$$W^*_\gamma(t, s) := \{w(U \cap (V_t\setminus X_t)) : U \in \mathcal{U}_\gamma \text{ and the trace of $U$ on $X_t$ is $s$} \}$$
Our goal is to ensure that each entry $D(t, s)$ contains an $(\ell, \epsilon_t)$-trimmed subset of $W^*_\gamma(t, s)$.
Additionally, the root $D(r, \{\gamma\}, 0, 1)$ is an $(\ell, \epsilon_t)$-trimmed subset of $W^*_\gamma(r, \{\gamma\}, 0,1 ) = W_\gamma^*$.

\paragraph{Dynamic programming transitions}
We compute the table $D$ using a post-order traversal of the tree decomposition (processing children before parents) with $\gamma$ added to all bags. We initialize all entries of $D$ to the empty set. We use the notation $Z \unionassign Z'$ to denote the update $Z \gets Z \cup Z'$.

\begin{enumerate}
\item If $t$ is a leaf with $X_t = \{\gamma\}$, the only valid partial solution consists solely of the center $\gamma$. We initialize the entry with zero weight: $D(t, \{\gamma\}, \rho(\gamma) = 0, \pi(\gamma) = 1) \unionassign \{0\}$.
\item If $t$ introduces a vertex $v^* \neq \gamma$ with child $t'$, we iterate over all trace $s$ for $t$. Let $s'$ be the restriction of $s$ for $t'$ by removing $v^*$. We update $D(t, s) \unionassign D(t', s')$ when $s$ satisfies one of the following conditions:
\begin{enumerate}
    \item $v^* \in P_0$ (i.e., $v^*$ is not selected).\label{proc:case_intro_vertex1}
    \item $\{v^*\} \in \partition$ forms a singleton component. In this case, we require $\pi(v^*) = 0$ (as the vertex has not found a parent yet) and $\rho(v^*) \in \{0, \dots, \radiusbd\}$.\label{proc:case_intro_vertex2}
\end{enumerate}
\item If $t$ forgets vertex $v^*$ with child $t'$, we verify the parent status and update the weights before removing the vertex. We iterate over every trace $s'$ for node $t'$ and project it to a trace $s$ for $t$ by removing $v^*$.
\begin{enumerate}
    \item If $v^* \in P_0'$ is not selected, we update $D(t, s) \unionassign D(t', s')$.\label{proc:case_forget_vertex1}
    \item If $v^*$ is selected with $v^* \in P_l'$ for some $l \neq 0$, we proceed only if $|P_l'| > 1$ and $\pi'(v^*) = 1$. Since no new edges will connect to $v^*$, we cannot forget a singleton component (as it would become permanently disconnected) and $v^*$ must already have a defined parent. If valid, we update:\label{proc:case_forget_vertex2}
$$D(t, s) \unionassign \{z' + w(v^*) : z' \in D(t', s')\}.$$
\end{enumerate}

\item If $t$ introduces an edge $(u^*, v^*)$ with child $t'$, we iterate over every trace $s'$ in the child table $D(t', \cdot)$ and generate new traces for $t$ based on the following cases:
\begin{enumerate}
    \item If at least one endpoint is not selected, $u^* \in P_0'$ or $v^* \in P_0'$. The edge is not part of the partial solution $G[U]$. The trace remains unchanged $s = s'$, and we update:
    $D(t, s) \unionassign D(t', s')$.\label{proc:case_add_edge1}
    \item When both endpoints are selected, $u^*, v^* \notin P_0'$, and the edge connects the two vertices. Let $\partition$ be the partition obtained by merging the components containing $u^*$ and $v^*$ in $\partition'$. We proceed only if the distance profile is consistent, and consider all valid updates to the parent status:\label{proc:case_add_edge2}
    \begin{itemize}
        \item Triangle Inequality Check: We require $|\rho'(u^*) - \rho'(v^*)| \le 1$. Otherwise, the trace is invalid because the edge contradicts the guessed shortest-path distances.
        \item Parent Status Update: We generate three $\pi$ from $\pi'$: $\pi = \pi'$ (No update), $\pi(u^*) = 1$ if $\rho'(u^*) = \rho'(v^*) + 1$ ($v^*$ is parent of $u^*$), and $\pi(v^*) = 1$ if $\rho'(v^*) = \rho'(u^*) + 1$ ($u^*$ is parent of $v^*$).  For each valid $\pi$, we update:
        $$D(t, \partition, \rho', \pi) \unionassign D(t', \partition', \rho', \pi').$$
    \end{itemize}
\end{enumerate}
\item \label{proc:case_join}
When vertex $t$ joins children $t_1$ and $t_2$, we iterate over all pairs of traces $s_1 = (\partition^1, \rho^1, \pi^1)$ for $t_1$ and $s_2 = (\partition^2, \rho^2, \pi^2)$ for $t_2$. A pair is \emph{compatible} if the set of selected vertices is identical ($P_0^1 = P_0^2 = P_0$) and the distance profiles match ($\rho^1 = \rho^2$).

For each compatible pair, we derive a merged trace $s = (\partition, \rho, \pi)$ as the following:
\begin{itemize}
    \item The partition $\partition$ is formed by the transitive closure of the connectivity in $\partition^1$ and $\partition^2$. Two vertices $u, v \in X_t \setminus P_0$ belong to the same component in $\partition$ if and only if they are connected by a sequence of edges existing within the components of $\partition^1$ or $\partition^2$.
    \item The distance profile is set to $\rho = \rho^1$ (which equals $\rho^2$), and the parent status become $\pi(u) = \max\{\pi^1(u), \pi^2(u)\}$ for all $u \in X_t \setminus P_0$. This ensures a vertex is marked as having a parent if it found one in either the left or right branch.
\end{itemize}
We update the table entry with the sum of weights:
$$D(t, s) \unionassign \{ z_1 + z_2 : z_1 \in D(t_1, s_1), z_2 \in D(t_2, s_2) \}.$$
\end{enumerate}

After computing each entry $D(t,s)$, we apply $\text{Trim}_{\ell,\delta}$ with $\delta = \epsilon/|V_\mathbb{T}|$ on $D$.  Specifically, we set a geometric grid $\{z\in \times_{i = 1}^d[0,R_i]: z_i \in \{0, 1, e^{\delta}, e^{2\delta},\dots\}\}$, and we only keep at most one weight within each bucket from $D(t,s)$ that maximizes the linear function $\ell$. As any two weights in the same bucket differ by at most a factor $e^{\delta}$ in all coordinates, keeping the $\ell$-maximizer ensures the result is an $(\ell, \delta)$-trimming of the original set.

At the root with $X_{root} = \{\gamma\}$, we set $W_\gamma\gets D(r, (\emptyset, \{\gamma\}),0,1)$, and return 
$W_{out} = \cup_{\gamma\in V} W_\gamma$ which union of these feasible weights with all possible choices of the center.

\paragraph{Correctness} 
We use induction on the tree decomposition to show that for every node $t$ and trace $s$, the DP table entry $D(t, s)$ contains an $(\ell, \epsilon_t)$-trimming of the feasible set $W^*_\gamma(t, s)$, where $\epsilon_t := \frac{|V_{\mathbb{T}_t}|}{|V_{\mathbb{T}}|}\epsilon$.

\begin{enumerate}
    \item When $t$ is a leaf, $X_t = \{\gamma\}$. The only valid partial solution is the singleton $U_t = \{\gamma\}$ with distance $\rho(\gamma)=0$ and parent status $\pi(\gamma)=1$ (by definition for the root of the solution). This satisfies the base case, because $U \cap (V_t\setminus X_t) = \emptyset$ for all $U\in \mathcal{U}_\gamma$.

    \item When introducing $v^*$, we iterate over all possible extensions of the child's trace $s'$. 
    If $v^*$ is not selected, the trace simply extends with $v^* \in P_0$. 
    If $v^*$ is selected (\Cref{proc:case_intro_vertex2}), we enumerate all possible distance guesses $r$ for $\rho(v^*)$ which ensure the completeness. We set $\pi(v^*) = 0$ since $v^*$ has no edges yet.
    \item 
    When forgetting $v^*$, we finalize its status and update the weights. Since $v^*$ is removed from the active bag, it cannot be adjacent to any vertices introduced in the future; thus, its parent status $\pi(v^*)$ and distance $\rho(v^*)$ are final. For completeness, consider any valid global solution $U \in \mathcal{U}_\gamma$ containing $v^*$. Because all neighbors of $v^*$ in $G[U]$ are in $V_{t'}$, the trace $s'$ for the child $t'$ must satisfy $\pi'(v^*) = 1$ and $\rho'(v^*) \le \radiusbd$. Moreover, $U$ is connected which ensures $|P_i'| > 1$. Therefore, the trace of $U$ is in the DP table. Finally, if $v^*$ is selected, we add its weight $w(v^*)$ to the DP value, ensuring that the weight of any selected vertex is counted exactly once.
    \item When introducing an edge $(u^*, v^*)$, if either endpoint is unselected (\Cref{proc:case_add_edge1}), we ignore the edge and copy the child's weights. If both are selected (\Cref{proc:case_add_edge2}), we merge their respective components. We then check the triangle inequality and enumerate all valid updates to the parent vector $\pi$. This ensures that if $(u^*, v^*)$ establishes the parent relationship for either $u^*$ or $v^*$ in the optimal solution, the corresponding trace is captured.
    \item  When joining children $t_1$ and $t_2$, we merge pairs of compatible traces sharing the same selected sets and distance profiles. For the approximation error, we sum the weights in sets from the children with $\epsilon_{t_1}$ and $\epsilon_{t_2}$, and then run $\text{Trim}_{\ell,\delta}$.  By \Cref{lem:trimming}, the resulting error bound satisfies $\epsilon_t \ge \epsilon_{t_1} + \epsilon_{t_2} + \delta$.
\end{enumerate}

Finally, at the root node $X_{root} = \{\gamma\}$, every selected vertex $v \neq \gamma$ is introduced and forgotten. By the validity check in~\Cref{proc:case_forget_vertex2}, $v$ has a parent $u$ with $\rho(u) = \rho(v) - 1$. By induction, this parent pointer chains back to the root $\gamma$ with $\rho(\gamma)=0$ which proves that a path of length $\rho(v)$ exists. Additionally, the triangle inequality check ensures that $\rho(v)$ corresponds to the true graph distance $d_{G[U]}(\gamma, v)$. Finally, since we forbid forgetting singleton components, every forgotten vertex is merged into the component containing the center $\gamma$. Thus, the final entry $D(root, \{\gamma\}, 0, 1)$ correctly represents the weights of connected subgraphs rooted at $\gamma$ with a strong radius of at most $\radiusbd$.

\paragraph{Running time}
The DP table is indexed by the tuple $(t, s)$. The tree decomposition has $O(n)$ nodes. For each node $t$, the number of valid partitions $\partition$ of the bag (size at most $tw+1$) is bounded by the Bell number, $2^{O(tw \log tw)}$. The distance profile $\rho$ has at most $(\radiusbd+1)^{tw+1}$ possibilities, and the parent status $\pi$ has $2^{tw+1}$ possibilities.

For any node $t$ and trace $s$, the number of weights $z$ is at most the size of the geometric grid, $O\left(\delta^{-d} \prod_{i=1}^d \ln R_i\right) = O\left(n^d \epsilon^{-d} \prod_{i=1}^d \ln R_i\right)$, as $\delta = \epsilon/|V_\mathbb{T}| = O(\epsilon/n)$.  Combining these factors, the maximum size of the table $D(t, \cdot)$ at any node $t$, denoted by $|D|_{max}$, is:
$$|D|_{max} = 2^{O(tw \log tw)} \cdot (\radiusbd+1)^{tw+1} \cdot 2^{tw+1} \cdot \left( \frac{n}{\epsilon} \right)^d \prod_{i=1}^d \ln R_i.$$

The running time of each DP transition is dominated by the join operation, which takes time quadratic in the table size, $|D|^2_{max}$.  Finally, as there are $O(n)$ nodes in the tree decomposition and we iterate over all $n$ possible centers $\gamma$, the total running time is 
$$O(n^2 \cdot |D|_{max}^2) = O\left( n^{2+2d} \epsilon^{-2d} (\radiusbd+1)^{2(tw+1)} 2^{O(tw \log tw)} \left(\prod_{i=1}^d \ln R_i\right)^2 \right).$$
\end{proof}

\subsection{Additional Lemmas}

\begin{proof}[Proof of \Cref{lem:separation_balanced2subgraphsum}]
Let $W^*$ be the set of connect subset sum, $W_1$ be a $(\ell_1, \varepsilon)$-trimmed $W^*$, and $W_2$ be a $(\ell_2, \varepsilon)$-trimmed $W^*$.  We prove that the final list $W_1\cup W_2$ contains a weakly violating $c$-balanced district $S$.

Let $w^* = w(S_\mathrm{max})\in \mathbb{Z}_{\ge 0}^4$. Without loss of generality, assume that $w_2^*\ge w_1^*$. Then, there exists a district $S\in W_1\subseteq W^*$ with $w' = w(S)$ that $\epsilon$-approximates and $\ell_1$-dominates $w^*$.  Therefore, $\ell_1(w') = (c-1)w_1' - w_2' \ge \ell_1(w^*) = (c-1)w^*_1 - w^*_2\ge 0$ by $\ell_1$-dominance.  Additionally, 
\begin{align*}
(c-1)w_2' - w_1' & \ge e^{-\varepsilon}(c-1)w_2^* - e^{\varepsilon}w_1^*\tag{$\varepsilon$-approximated}\\
&\ge \left(e^{-\varepsilon}(c-1) - e^{\varepsilon}\right)w_1^* \tag{$w_2^*\ge w_1^*$}\\
&\ge 0 \tag{whenever $e^{2\varepsilon}<(c-1)$}
\end{align*}
The above two show that $S$ is $c$-balanced. As $w_4'/w_4^*\in [e^{-\varepsilon}, e^{\varepsilon}]$, we can show that $S$ is a weakly violating constraint:
\begin{align*}
w_4'&\le e^{\varepsilon}\cdot w_4^*\tag{$\varepsilon$-approximated}\\
&\le e^{\varepsilon}\cdot (1-2\epsilon) \cdot  w_3^* \tag{$S_\mathrm{max}$ is strongly violating}\\
&\le e^{\varepsilon}\cdot (1-2\epsilon) \cdot e^\varepsilon w_3'\tag{$\varepsilon$-approximated}\\
&\le e^{2\varepsilon} \cdot (1-2\epsilon) \cdot w_3' \le (1-\epsilon) \cdot w'_3. \tag{small enough $\varepsilon > 0$}
\end{align*}
Finally, $w_3'\ge e^{-\varepsilon}\cdot w_3^*
\ge \frac{1}{2} w_3^*$,
certifying that the output $S$ satisfies all the constraints from the lemma statement.
\end{proof}

The following lemma shows that the parameter $\epsilon$ decay smoothly under composition and addition.
\begin{lemma}\label{lem:trimming}
    Given $\epsilon, \epsilon_1, \epsilon_2\ge 0$ and $\ell$, if $W_1'$ is an $(\ell, \epsilon_1)$-trimmed of $W_1$ and $W_2'$ is an $(\ell, \epsilon_2)$-trimmed of $W_2$, if $W''$ is an $(\ell, \epsilon)$-trimmed of $W_1'+W_2':= \{z_1'+z_2': z_1'\in W_1', z_2'\in W_2'\}$, $W''$ is an $(\ell, \epsilon+\max (\epsilon_1, \epsilon_2))$-trimmed of $W_1+W_2$.
\end{lemma}
\begin{proof}[Proof of \Cref{lem:trimming}]
    We first show $W_1'+W_2'$ is a $(\ell, \max (\epsilon_1, \epsilon_2))$-trimmed of $W_1+W_2$.  $W_1'+W_2'\subseteq W_1+W_2$ is trivial.  For any $z_1 = (z_{1,1}, \dots, z_{1,d})\in W_1$ and $z_2 = (z_{2,1}, \dots, z_{2,d})\in W_2$, there exists $z_1'\in W_1'$ and $z_2'\in W_2'$ so that 
    $$\ell(z_1')\ge \ell(z_1), \ell(z_2')\ge \ell(z_2)\text{, and }\frac{z_{1,i}}{z_{1,i}'}, \frac{z_{2,i}}{z_{2,i}'}\in [e^{-\epsilon_1}, e^{\epsilon_1}]\text{ for all }i = 1,\dots,d.$$
    Hence, $z_1'+z_2'\in W_1'+W_2'$ and $z_1+z_2\in W_1+W_2$ satisfy 
    $$\ell(z_1'+z_2')\ge \ell(z_1+z_2)\text{ and }\frac{z_{1,i}+z_{2,i}}{z_{1,i}'+z_{2,i}'}\in [e^{-\max (\epsilon_1, \epsilon_2)}, e^{\max (\epsilon_1, \epsilon_2)}]$$
    If $W''$ is an $(\ell, \epsilon)$-trimmed of $W_1'+W_2'$, then $W''\subseteq W_1'+W_2'\subseteq W_1+W_2$, and by a similar argument we have $W''$ is $(\ell, \epsilon+\max(\epsilon_1, \epsilon_2))$-trimmed of $W_1+W_2$.
\end{proof}
\section{Proof of \Cref{cor:ptas_relax_bounded}}
\label{app:multisubsetsum}

We use dynamic programming to find the weights of vertex-disjoint districts. Similar to \cref{lem:allsubsetsum}, for relaxed composition constraints, we only need to compute an approximation of all partial districtings.




For each node $t \in V_\mathbb{T}$, we classify a district $U$ as \emph{open} if it intersects the current bag $X_t$, and closed if it is entirely contained in the already processed subgraph $V_t \setminus X_t$.  Unlike~\Cref{lem:allsubsetsum} tracking a single open district, we need to simultaneously maintain the state (or trace) of multiple open districts along with the total weight of all closed districts. Accordingly, each DP state is indexed by node $t$, and the {combined traces} of the open districts store the list of \emph{weight profiles}.

\subsection{Strong Radius Constraints}
\paragraph{Dynamic programming states}  The DP table $D$ is indexed by a pair $(t, \vs)$, where the \emph{combined trace} $\vs = (s_1, \dots, s_I)$ represents $I$ open districts $U_1,\dots,U_I$. Each individual trace $s_i = (\mathbb{P}_i, \rho_i, \pi_i, \gamma_i)$ with a center $\gamma_i$ similar to~\Cref{lem:allsubsetsum}:

For each open district $U_i$ and bag $X_t$, the trace $s_i$ includes a partition $\partition_i = (P_{i,1},\dots,P_{i,L_i})$ of disjoint subsets of $X_t$, where each $P_{i,l}$ corresponds to the intersection of the bag with a connected component of $U_i$ in the current subgraph $G_t[U_i]$. We denote the set selected vertices in the bag by $U_i$ as $P_i = \cup_{l} P_{i,l}$, and the set of the unselected as $P_0 = X_t \setminus \cup_{i} P_i$.  

To verify compactness and connectivity, we maintain a distance profile $\rho_i: P_i \to \{0, \dots, \radiusbd\}$ relative to a center $\gamma_i$, where $\rho_i(u)$ is the guessed shortest-path distance $d_{G[U_i]}(u,\gamma_i)$ to the center.  Additionally, we track a parent status $\pi_i: P_i \to \{0, 1\}$, where $\pi_i(u) = 1$ indicates if it has a parent $v$ in the partial solution with $\rho_i(u) = \rho_i(v) + 1$.

Given a valid districting solution $\mathcal{T} = \{U_1, \dots\}$ and a node $t$ in the tree decomposition, let $U_1, \dots, U_I$ denote the open districts relative to $t$, and let $\mathcal{T}_{closed} \subseteq \mathcal{T}$ denote the collection of closed districts. The \emph{weight profile} of $\mathcal{T}$ on the bag $X_t$ is 
$$w(\mathcal{T}|t) :=  \left(w(U_1\cap (V_t\setminus X_t)),\dots,w(U_I \cap (V_t\setminus X_t)),\sum_{U\in \mathcal{T}_{closed}}w(U)\right)\in \mathbb{Z}^{(I+1)\times d}$$ which aggregates the weights for the open districts and the sum of closed districts.  Moreover, the combined trace of $\mathcal{T}$ on the bag $X_t$ is $\vs = (s_1,\dots,s_I)$ so that the trace of $U_i$ on bag $X_t$ is $s_i$ with center $\gamma_i$ as defined in \Cref{app:allsubsetsum}.  Conversely, $\mathcal{T}$ is \emph{compatible} with $\vs$, if the trace of $\mathcal{T}$ is $\vs$.

Finally, we define the feasible weight profile set $W^\oplus(t, \vs)$ as the collection of weight profiles for compatible districting: 
$$W^\oplus(t, \vs) := \left\{ w(\mathcal{T}|t) : \text{$\mathcal{T}$ is valid districting and compatible with $\vs$ on bag $X_t$}\right\}$$

\paragraph{Dynamic programming transitions}
We describe the following procedure to update the table $D$ in the depth first search ordering for all node $t\in V_\mathbb{T}$.  The major difference is the forget node, while the rest is similar to \Cref{lem:allsubsetsum}. Recall that we use the notation $Z \unionassign Z'$ to denote the update $Z \gets Z \cup Z'$.

\begin{enumerate}
\item If $t$ is a leaf ($X_t = \emptyset$), we set the empty state with zero weight,$D(t, \emptyset) \unionassign \{0\}$
\item If $t$ introduces a vertex $v^*$ with child $t'$, we iterate over all $\vs$ for $t$. Let $\vs'$ be the restriction of $\vs$ for $t'$ by removing $v^*$. We update $D(t, \vs) \unionassign D(t', \vs')$ when $\vs$ satisfies one of the following conditions:
\begin{enumerate}
    \item $v^* \in P_0$ is not selected.
    \item When $v^*\in P_i$ is selected with $i\neq 0$ and $\{v^*\} \in P_{i,l}$ forms a singleton component, if $v^*\neq \gamma_i$, we require $\pi_i(v^*) = 0$ and $\rho_i(v^*) \in \{0, \dots, \radiusbd\}$. similar to~\Cref{proc:case_intro_vertex2}.  If $v^* = \gamma_i$, we require $\pi_i(v^*) = 1$ and $\rho_i(v^*) = 0$.
\end{enumerate}
\item 
If $t$ forgets vertex $v^*$ with child $t'$, we verify the parent status and update the weights before removing the vertex. We iterate over every trace $\vs'$ for $t'$ and project it to a trace $\vs$ for $t$ by removing $v^*$.

\begin{enumerate}
    \item If $v^*$ is not selected ($v^* \in P'_0$), we update 
    $D(t, \vs) \unionassign D(t', \vs')$ similar to~\Cref{proc:case_forget_vertex1} in~\cref{app:allsubsetsum}.
    \item If $v^*$ is selected in district $i^*$ (i.e., $v^*$ belongs to a component in $\partition'_{i^*}$), we first verify $\pi'_{i^*}(v^*) = 1$. Unlike~\Cref{proc:case_forget_vertex2} in~\cref{app:allsubsetsum}, as we have multiple open districts, consider the following three cases based on the structure of district $i^*$:
    \begin{itemize}
        \item If $\{v^*\} \in \partition'_{i^*}$ is a singleton component but the district has other disjoint components in the bag, $|\partition'_{i^*}| > 1$, we discard the state, because we cannot forget a component that has not yet merged with the rest of its district.
        \item If $v^* \in P'_{i^*,j}$ with $|P'_{i^*,j}| > 1$, we update the weights of district $i^*$
        $$D(t, \vs) \unionassign \{ \vz + w(v^*) \cdot \mathbf{e}_{i^*} : \vz \in D(t', \vs')\}$$
        which updates the weight of $i^*$-th open district.
        \item If $\{v^*\}$ is the only component of district $i$ in the bag, $\partition'_{i^*} = \{ \{v^*\} \}$, we first check the $\delta$-relaxed composition constraints~\cref{eq:def_relax_balanced,eq:def_relax_threshold}.  Once it's valid, we close the district and move the accumulated weight $z'_{i^*}$ to the closed accumulator $z_0$.
        $$D(t, \vs) \unionassign \{ \vz : z_0 = z'_0 + z'_{i^*}+w(v^*), z_{i^*} = 0, \text{ and }z_{\iota} = z'_\iota ,\forall \iota\neq i^*, \vz' \in D(t', \vs')\}$$
    \end{itemize}
\end{enumerate}

\item When $t$ introduces edge $(u^*,v^*)$ and has child $t'$, we merge the partition, verify the distance profile, and update the parent status.  For any state $\vs'$ in $D(t',\cdot)$,
\begin{enumerate}
    \item If $u^* \in P'_i$ and $v^* \in P'_{i'}$ with $i \neq i'$ or $i = i' = 0$, we simply ignore the edge and $D(t,\vs')\unionassign D(t',\vs')$.
    \item If $u^*, v^* \in P'_{i^*}$ are in the same district, we run the same procedure as \Cref{proc:case_add_edge2} in~\cref{app:allsubsetsum} to create $\vs$ and update $D(t, \vs) \unionassign D(t',\vs')$.
\end{enumerate}
\item When $t$ joins children $t_1$ and $t_2$, we iterate over compatible pairs $\vs^1, \vs^2$ in $D(t_1,\cdot)$ and $D(t_2,\cdot)$ respectively so that $P^1_i = P^2_i$, $\rho^1_i = \rho^2_i$, and $\gamma^1_i = \gamma^2_i$ for all $i$ and update $D(t, \vs)$ as \Cref{proc:case_join}.
\end{enumerate}

After computing the transitions for $D(t, \vs)$ at each node $t$, we apply a trimming step with $\delta_{alg} = \min\{\epsilon,\delta/3\}/|V_\mathbb{T}|$ to keep the size of the dynamic programming table polynomial as \Cref{app:allsubsetsum}.  The only difference is we can pick an arbitrary weight profile in each bucket.  At the root of the tree decomposition, $X_{root} = \emptyset$, and all valid districts are closed.  We output $D(root, \emptyset)$.

\paragraph{Correctness} We use induction to show that the DP table entry $D(t, \vs)$ contains an $\epsilon_t$-approximation of the feasible set $W^\oplus(t, \vs)$ with $\epsilon_t := \frac{|V_{\mathbb{T}_t}|}{|V_{\mathbb{T}}|}\epsilon$. 
The proof is similar to \Cref{lem:allsubsetsum}, and we only prove the approximation guarantee for the case of forgetting a vertex.

The structural checks (ensuring radius $\rho_i \le \radiusbd$, parent tracking $\pi_i = 1$, and proper connectivity of forgotten components) ensure no invalid districts are propagated or closed prematurely. By the induction hypothesis, the child node $t'$ contains an $\epsilon_{t'}$-approximation. The trimming step at node $t$ introduces an additional multiplicative error of at most $\delta$. Thus, the combined error is bounded by $\epsilon_{t'} + \delta$. Since $|V_{\mathbb{T}_t}| = |V_{\mathbb{T}_{t'}}| + 1$, we have $\epsilon_{t'} + \delta \le \epsilon_t$. For the forget node, closing a fully formed district shifts its accumulated weight to the closed districts accumulator $z_0$, preserving the multiplicative error.

Finally, the DP evaluates the relaxed composition constraints ($c$-balancedness or $B$-threshold) at the forget node $t$. Because the DP maintains an $\epsilon$-approximation, for any district $U_{i^*}$ satisfying the exact composition constraints and close on node $t$, the DP the tracked weight vector $z$ of a district approximates its true weight $w^* = w(U_{i_*})$.  Therefore, for $B$-threshold districts: The exact optimal solution requires $w^*_1 \ge B$, so its approximate tracked weight satisfies $z_{1} \ge e^{-\delta_{alg}}B \ge e^{-\delta}B\ge (1-\delta)B$ for small enough $\delta$, which passes the check. Similar argument holds for $c$-balanced districts.  This yields an $\epsilon$-approximation to the optimal weight satisfying the $\delta$-relaxed constraints.

\paragraph{Running time}

For a node $t$ in the tree decomposition, the bag size is $|X_t| \le tw+1$. The number of open districts $I$ is at most $tw+1$. To bound the number of states $\vs = (s_1, \dots, s_I)$, we first characterize the choice of $\partition_1,\dots,\partition_I$.  Here we first partition the bag into connected components, and then group them into open districts. Let $S(p,b)$ be the number of ways to partition $p$ items into $b$ non-empty parts (known as the Stirling number of the second kind).  First, we partition the bag into exactly $b$ non-empty parts. Second, we select one part to be the unselected set $P_0$ (or none).  Finally, we group the remaining $b-1$ parts into $c \le b-1$ open districts.  Summing over all possible values of $b$ and $c$, the choice of $\partition_1,\dots,\partition_I$ is bounded by 
$$\sum_{b=1}^{tw+1} S(tw+1,b) \left(\sum_{c=1}^{b-1} b\cdot S(b-1,c) \right)\le (tw+1)\left(\sum_{b=1}^{tw+1} S(tw+1,b)\right)^2 = 2^{O(tw \ln tw)},$$
because the sum $\sum_{b=1}^{tw+1} S(tw+1, b)$ is the number of partition (known as Bell number) is in $2^{O(tw \log tw)}$. 

The number of distance profile and parent status pair is in $O((2(\radiusbd+1))^{tw+1})$, and the number of centers is $n^{tw+1}$  Thus, the total number of valid combined trace $\vs$ at any node $t$ is bounded by $n^{tw+1}2^{O(tw \log tw)} (2(\radiusbd + 1))^{tw+1}$.  For each state $(t,\vs)$, the trimming step bounds the number of weight profiles in $D(t, \vs)$ to the number of buckets $O\left(\left(\frac{n \ln R}{\epsilon}\right)^{(tw+2)d}\right)$ as the dimension of weight profile $\vz$ is at most $d(I+1)\le d(tw+2)$.  Combining these factors, the maximum size of the table $D(t, \cdot)$ at any node $t$, denoted by $|D|_{max}$, is:
$$|D|_{max} = 2^{O(tw \log tw)} \cdot (\radiusbd+1)^{tw+1} \cdot n^{tw+1} \cdot \left( \frac{n}{\epsilon} \right)^{(tw+2)d} \left(\prod_{i=1}^d \ln R_i\right)^{tw+2}.$$
The running time of each DP transition is dominated by the join operation, which takes time quadratic in the table size, $|D|^2_{max}$.  Finally, as there are $O(n)$ nodes in the tree decomposition, the total running time is 
$$O(n \cdot |D|_{max}^2) = O\left(n^{(2d+1)(tw+2)} \epsilon^{-2d(tw+2)} (\radiusbd+1)^{2(tw+1)} 2^{O(tw \log tw)} \left(\prod_{i=1}^d \ln R_i\right)^{2(tw+2)} \right).$$

\subsection{Weak Radius Constraints}
For the weak radius case, the algorithm uses the same dynamic programming framework as in the strong radius case, but with a simpler state space. In the strong radius setting, distances are defined within the induced subgraph and must be tracked dynamically via distance profiles and parent statuses.  However, the weak radius constraint depends on the static distance on the original graph $G$.  Therefore, we can precompute them using All-Pairs Shortest Paths.

We only need to modify the \emph{Introduce Node} step.  When considering the addition of a vertex $v$ to an open district $i$ centered at $\gamma_i$, we verify 
$$d_G(v, \gamma_i) \le \radiusbd$$
If the condition holds, we proceed with the standard partition update; otherwise, the branch is pruned. 
By eliminating the distance profile $\rho$ and parent status $\pi$ from the state, the number of trace for each node is reduced by a $(k+1)^{tw+1}$ factor.

\end{document}